\documentclass[aps,11pt,twoside, nofootinbib]{revtex4}
\usepackage{amsthm}
\usepackage{amsmath,latexsym,amssymb,verbatim,enumerate,graphicx}
\usepackage{hyperref}
\usepackage{color}

\newtheorem{thm}{Theorem}
\newtheorem*{thm*}{Theorem}
\newtheorem{definition}[thm]{Definition} %[section]
\newtheorem{prop}[thm]{Proposition}
\newtheorem*{prop*}{Proposition}
\newtheorem{lem}[thm]{Lemma}
\newtheorem*{lem*}{Lemma}

\newtheorem*{fact*}{Fact}
\newtheorem{cor}[thm]{Corollary}
\newtheorem*{cor*}{Corollary}

\newtheorem*{rep@theorem}{\rep@title}
\newcommand{\newreptheorem}[2]{%
\newenvironment{rep#1}[1]{%
 \def\rep@title{#2 \ref{##1} (restatement)}%
 \begin{rep@theorem}}%
 {\end{rep@theorem}}}
\makeatother

\newreptheorem{thm}{Theorem}
\newreptheorem{lem}{Lemma}
\newreptheorem{cor}{Corollary}

\usepackage{prettyref}
\newcommand{\savehyperref}[2]{\texorpdfstring{\hyperref[#1]{#2}}{#2}}

\newrefformat{eq}{\savehyperref{#1}{\textup{(\ref*{#1})}}}
\newrefformat{lem}{\savehyperref{#1}{Lemma~\ref*{#1}}}
\newrefformat{def}{\savehyperref{#1}{Definition~\ref*{#1}}}
\newrefformat{thm}{\savehyperref{#1}{Theorem~\ref*{#1}}}
\newrefformat{cor}{\savehyperref{#1}{Corollary~\ref*{#1}}}
\newrefformat{cha}{\savehyperref{#1}{Chapter~\ref*{#1}}}
\newrefformat{sec}{\savehyperref{#1}{Section~\ref*{#1}}}
\newrefformat{app}{\savehyperref{#1}{Appendix~\ref*{#1}}}
\newrefformat{tab}{\savehyperref{#1}{Table~\ref*{#1}}}
\newrefformat{fig}{\savehyperref{#1}{Figure~\ref*{#1}}}
\newrefformat{hyp}{\savehyperref{#1}{Hypothesis~\ref*{#1}}}
\newrefformat{alg}{\savehyperref{#1}{Algorithm~\ref*{#1}}}
\newrefformat{rem}{\savehyperref{#1}{Remark~\ref*{#1}}}
\newrefformat{item}{\savehyperref{#1}{Item~\ref*{#1}}}
\newrefformat{step}{\savehyperref{#1}{step~\ref*{#1}}}
\newrefformat{conj}{\savehyperref{#1}{Conjecture~\ref*{#1}}}
\newrefformat{fact}{\savehyperref{#1}{Fact~\ref*{#1}}}
\newrefformat{prop}{\savehyperref{#1}{Proposition~\ref*{#1}}}
\newrefformat{prob}{\savehyperref{#1}{Problem~\ref*{#1}}}
\newrefformat{claim}{\savehyperref{#1}{Claim~\ref*{#1}}}
\newrefformat{relax}{\savehyperref{#1}{Relaxation~\ref*{#1}}}
\newrefformat{red}{\savehyperref{#1}{Reduction~\ref*{#1}}}
\newrefformat{part}{\savehyperref{#1}{Part~\ref*{#1}}}

\newcommand{\MYstore}[2]{%
  \global\expandafter \def \csname MYMEMORY #1 \endcsname{#2}%
}

\newcommand{\MYload}[1]{%
  \csname MYMEMORY #1 \endcsname%
}

\newcommand{\MYnewlabel}[1]{%
  \newcommand\MYcurrentlabel{#1}%
  \MYoldlabel{#1}%
}

\newcommand{\MYdummylabel}[1]{}

\newcommand{\torestate}[1]{%
  \let\MYoldlabel\label%
  \let\label\MYnewlabel%
  #1%
  \MYstore{\MYcurrentlabel}{#1}%
  \let\label\MYoldlabel%
}

\newcommand{\restatethm}[1]{%
  \let\MYoldlabel\label
  \let\label\MYdummylabel
  \begin{repthm}{#1}
    \MYload{#1}
  \end{repthm}
  \let\label\MYoldlabel
}

\newcommand{\restatelem}[1]{%
  \let\MYoldlabel\label
  \let\label\MYdummylabel
  \begin{replem}{#1}
    \MYload{#1}
  \end{replem}
  \let\label\MYoldlabel
}

\newcommand{\restatelemalt}[1]{%
  \let\MYoldlabel\label
  \let\label\MYdummylabel
  \begin{lem*}[Restatement of \prettyref{#1}]
    \MYload{#1}
  \end{lem*}
  \let\label\MYoldlabel
}

\newcommand{\restatecor}[1]{%
  \let\MYoldlabel\label
  \let\label\MYdummylabel
  \begin{repcor}{#1}
    \MYload{#1}
  \end{repcor}
  \let\label\MYoldlabel
}

\newcommand{\restateprop}[1]{%
  \let\MYoldlabel\label
  \let\label\MYdummylabel
  \begin{prop*}[Restatement of \prettyref{#1}]
    \MYload{#1}
  \end{prop*}
  \let\label\MYoldlabel
}

\newcommand{\restatefact}[1]{%
  \let\MYoldlabel\label
  \let\label\MYdummylabel
  \begin{fact*}[Restatement of \prettyref{#1}]
    \MYload{#1}
  \end{fact*}
  \let\label\MYoldlabel
}

\newcommand{\restate}[1]{%
  \let\MYoldlabel\label
  \let\label\MYdummylabel
  \MYload{#1}
  \let\label\MYoldlabel
}

\def\ba#1\ea{\begin{align}#1\end{align}}
\def\ban#1\ean{\begin{align*}#1\end{align*}}

\newcommand{\ot}{\otimes}
\newcommand{\be}{\begin{equation}}
\newcommand{\ee}{\end{equation}}
\def\hcal{{\cal H}}

\def\benum{\begin{enumerate}}
\def\eenum{\end{enumerate}}

\def\squareforqed{\hbox{\rlap{$\sqcap$}$\sqcup$}}
\def\qed{\ifmmode\squareforqed\else{\unskip\nobreak\hfil
\penalty50\hskip1em\null\nobreak\hfil\squareforqed
\parfillskip=0pt\finalhyphendemerits=0\endgraf}\fi}
\def\endenv{\ifmmode\;\else{\unskip\nobreak\hfil
\penalty50\hskip1em\null\nobreak\hfil\;
\parfillskip=0pt\finalhyphendemerits=0\endgraf}\fi}
\newcommand{\bra}[1]{\langle #1|}
\newcommand{\ket}[1]{|#1\rangle}
\newcommand{\braket}[2]{\langle #1|#2\rangle}
\DeclareMathOperator{\tr}{tr}

\newcommand{\bbC}{\mathbb{C}}
\newcommand{\bbU}{\mathbb{U}}
\newcommand{\id}{\mathbb{I}}

\newcommand{\ben}{\begin{equation}}
\newcommand{\een}{\end{equation}}

\newcommand{\<}{\langle}
\renewcommand{\>}{\rangle}
\newcommand{\I}{{\rm I}}
\def\L{\left}
\def\R{\right}

\def\id{{\operatorname{id}}}
\DeclareMathOperator{\ad}{ad}
\DeclareMathOperator{\diag}{diag}
\DeclareMathOperator{\sym}{sym}
\DeclareMathOperator{\rank}{rank}
\DeclareMathOperator{\Par}{Par}
\DeclareMathOperator{\supp}{supp}
\def\be{\begin{equation}}
\def\ee{\end{equation}}
\def\ben{\begin{eqnarray}}
\def\een{\end{eqnarray}}
\def\ot{\otimes}

\def\bei{\begin{itemize}}
\def\eei{\end{itemize}}

\def\E{{\mathbb{E}}}

\def\F{{\mathbb{F}}}
\def\I{{\mathbb{I}}}

\def\ep{\epsilon}
\def\eps{\epsilon}

\def\hcal{{\cal H}}
\def\gcal{{\cal G}}
\def\cN{{\cal N}}
\def\cP{{\cal P}}
\def\cQ{{\cal Q}}
\def\cS{{\cal S}}

\mathchardef\ordinarycolon\mathcode`\:
\mathcode`\:=\string"8000
\def\vcentcolon{\mathrel{\mathop\ordinarycolon}}
\begingroup \catcode`\:=\active
  \lowercase{\endgroup
  \let :\vcentcolon
  }

\newcommand{\nc}{\newcommand}
\nc{\rnc}{\renewcommand} \nc{\beq}{\begin{equation}}
\nc{\eeq}{{\end{equation}}} \nc{\bea}{\begin{eqnarray}}
\nc{\eea}{\end{eqnarray}} \nc{\beqa}{\begin{eqnarray}}
\nc{\eeqa}{\end{eqnarray}} \nc{\lbar}[1]{\overline{#1}}
 \nc{\proj}[1]{|#1\rangle\!\langle #1 |} 
\nc{\avg}[1]{\langle#1\rangle}

\nc{\conv}{\operatorname{conv}}
\nc{\smfrac}[2]{\mbox{$\frac{#1}{#2}$}} \nc{\Tr}{\operatorname{Tr}}
\nc{\ox}{\otimes} \nc{\dg}{\dagger} \nc{\dn}{\downarrow}
\nc{\lmax}{\lambda_{\text{max}}}
\nc{\lmin}{\lambda_{\text{min}}}

\nc{\csupp}{{\operatorname{csupp}}}
\nc{\qsupp}{{\operatorname{qsupp}}} \nc{\var}{\operatorname{var}}
\nc{\rar}{\rightarrow} \nc{\lrar}{\longrightarrow}
\nc{\poly}{\operatorname{poly}}
\nc{\polylog}{\operatorname{polylog}} \nc{\Lip}{\operatorname{Lip}}
\nc{\Om}{\Omega}
\nc{\wt}[1]{\widetilde{#1}}

\def\>{\rangle}
\def\<{\langle}

\def\d{\delta}

\nc{\glneq}{{\raisebox{0.6ex}{$>$}  \hspace*{-1.8ex} \raisebox{-0.6ex}{$<$}}}
\nc{\gleq}{{\raisebox{0.6ex}{$\geq$}\hspace*{-1.8ex} \raisebox{-0.6ex}{$\leq$}}}

\nc{\vholder}[1]{\rule{0pt}{#1}}
\nc{\wh}[1]{\widehat{#1}}
\nc{\h}[1]{\widehat{#1}}

\nc{\ob}[1]{#1}

\def\beq{\begin {equation}}
\def\eeq{\end {equation}}

\def\be{\begin{equation}}
\def\ee{\end{equation}}

\nc{\eq}[1]{(\ref{eq:#1})} 
\nc{\eqs}[2]{\eq{#1} and \eq{#2}}

\nc{\eqn}[1]{Eq.~(\ref{eqn:#1})}
\nc{\eqns}[2]{Eqs.~(\ref{eqn:#1}) and (\ref{eqn:#2})}

\newcommand{\secref}[1]{Section~\ref{sec:#1}}
\newcommand{\appref}[1]{Appendix~\ref{sec:#1}}
\newcommand{\lemref}[1]{Lemma~\ref{lem:#1}}
\newcommand{\thmref}[1]{Theorem~\ref{thm:#1}}
\newcommand{\propref}[1]{Proposition~\ref{prop:#1}}

\newcommand{\defref}[1]{Definition~\ref{def:#1}}
\newcommand{\corref}[1]{Corollary~\ref{cor:#1}}

\nc{\region}{\cS\cW}

\begin{document}

\title{{\Large Local random quantum circuits are approximate polynomial-designs}}

\author{Fernando G.S.L. Brand\~ao}
\email{fgslbrandao@gmail.com}
\affiliation{Department of Computer Science, University College London, London, UK}
%\affiliation{Institute for Theoretical Physics, ETH Z\"urich, 8093 Z\"urich, Switzerland}
%\affiliation{Centre for Quantum Technologies, National University of Singapore, 2 Science Drive 3, 117543 Singapore} 

\author{Aram W. Harrow}
\email{aram@mit.edu}
\affiliation{Center for Theoretical Physics, Massachusetts Institute
  of Technology, Cambridge, MA, USA}

\author{Micha\l{} Horodecki}
 \email{fizmh@ug.edu.pl}
\affiliation{Institute for Theoretical Physics and Astrophysics, University of Gda\'nsk, 80-952 Gda\'nsk, Poland}

\begin{abstract}

We prove that local random quantum circuits acting on $n$ qubits composed of $ O(t^{10} n^2)$ many nearest neighbor two-qubit gates form an approximate unitary $t$-design. Previously it was unknown whether random quantum circuits were a $t$-design for any $t > 3$.

The proof is based on an interplay of techniques from quantum many-body theory, representation theory, and the theory of Markov chains. In particular we employ a result of Nachtergaele for lower bounding the spectral gap of frustration-free quantum local Hamiltonians; a quasi-orthogonality property of permutation matrices; a result of Oliveira which extends to the unitary group the path-coupling method for bounding the mixing time of random walks; and a result of Bourgain and Gamburd showing that dense subgroups of the special unitary group, composed of elements with algebraic entries, are $\infty$-copy tensor-product expanders.

We also consider pseudo-randomness properties of local random quantum circuits of small depth and prove that circuits of depth $ O(t^{10}n)$ constitute a quantum $t$-copy tensor-product expander. The proof also rests on techniques from quantum many-body theory, in particular on the detectability lemma of Aharonov, Arad, Landau, and Vazirani. 

We give applications of the results to cryptography, equilibration of closed quantum dynamics, and the generation of topological order. In particular we show the following pseudo-randomness property of generic quantum circuits:
Almost every circuit $U$ of size $O(n^{k})$ on $n$ qubits cannot be distinguished from a Haar uniform unitary by circuits of size
$O(n^{(k-9)/11})$ that are given oracle access to $U$.

%; this provides a data-hiding scheme against computationally bounded adversaries. Second we reconsider a recent argument of Masanes, Roncaglia, and Acin concerning local equilibration of time-evolving quantum systems, and strengthen the connection between fast equilibration of small subsystems and the circuit complexity of the unitary which diagonalizes the Hamiltonian. Third we show that in one dimension almost every parallel local circuit of linear depth generates topological order, matching an upper bound to the problem due to Bravyi, Hastings, and Verstraete. 
\end{abstract}

\maketitle

\parskip .75ex

%%%%%%%%%%%%%%%%%%%%%%%%%%%%%%%%%%%%%%%%%%%%%%%%%%%%%%%%%%%%%%%%%%%%%%%%

\section{Introduction}

Random unitary matrices are an important resource in quantum information theory and quantum computing. Examples of the use of random unitaries, drawn from the Haar measure on the unitary group, include the encoding for almost every known protocol for sending information down a quantum channel \cite{ADHW06}, approximate encryption of quantum information \cite{HLSW04}, quantum data-hiding \cite{HLSW04}, information locking \cite{HLSW04}, and solving certain instances of the hidden subgroup problem over non-abelian groups \cite{Sen06}. Yet random unitary matrices are unreasonable from a computational point of view: To implement a random Haar unitary one needs an exponential number of two-qubit gates and random bits \cite{Knill95}. Thus it is interesting to explore constructions of \textit{pseudo}-random unitaries, which can be efficiently implementable and can replace random unitaries in some respects.

An approximate unitary $t$-design is a distribution of unitaries which
mimic properties of the Haar measure for polynomials of degree up to
$t$ (in the entries of the unitaries) \cite{TothGR07,DCEL09, GAE07, ELL05,
  ODP06, DOP07, HL09, DJ10, AB08, HP07, Zni08, HL08, RS09, BV10,
  BH10}. Approximate designs have a number of interesting applications
in quantum information theory replacing the use of truly random
unitaries (see e.g. \cite{EWS03, ODP06, DOP07, HP07, BH10, HH08,
  Low09}). It has been a conjecture in the theory of quantum
pseudo-randomness that polynomial-size random quantum circuits on $n$
qubits form an approximate unitary $\poly(n)$-design
\cite{HL09}. Analogously, polynomial-size {\em reversible} circuits
are known to form approximately $\poly(n)$-wise independent
permutations \cite{BH08} (see also \cite{KNR09}). 
However, up to now, the best result known for quantum circuits was that
polynomial random quantum circuits are approximate unitary $3$-designs
\cite{BH10}, which improved on a series of papers establishing that
random circuits are approximate unitary $2$-designs \cite{ODP06, DOP07,
  HL09, DJ10, AB08, Zni08}. Moreover, efficient constructions of
quantum $t$-designs, using a polynomial number of quantum gates and
random bits, were only known for $t = O(n/ \log(n))$ \cite{HL08}. In
this paper we make progress in the problem of unitary $t$-designs. We
prove that local random quantum circuits acting on $n$ qubits composed
of polynomially many nearest neighbour two-qubit gates form an
approximate unitary $\poly(n)$-design, settling the conjecture in the
affirmative. An important conceptual advance, which we will build on in this paper,
was the realization that one can connect the problem of showing that random circuits
are $t$-unitary designs to lower bounding the spectral gap of a many-body quantum
Hamiltonian \cite{BV10}.

In the remainder of this section, we will give the definitions and
notation used in this paper.  Then we will state the main result in
\secref{main} and outline a few applications in \secref{app}.  The rest of the paper is
devoted to the proof, with an overview in \secref{proof-OV} and
the details in \secref{proof-details}.

\subsection{Approximate Unitary Designs and Quantum Tensor-Product Expanders}

We start with the definition of tensor-product expanders \cite{HH09},
which are objects similar to approximate unitary designs, but with the
approximation to the Haar measure quantified differently. Let $\mu_{\text{Haar}}$ be the Haar measure 
on $\mathbb{U}(N)$ (the group of $N \times N$ unitary matrices).

\begin{definition}\label{def:TPE}
Let $\nu$ be a distribution on $\mathbb{U}(N)$. Then $\nu$ is a $(N, \lambda, t)$ quantum
$t$-copy tensor-product expander (or TPE for short) if  
\be
g(\nu,t) := 
\left \Vert \int_{\mathbb{U}(N)} U^{\otimes t, t} \nu({\rm d}U) -
  \int_{\mathbb{U}(N)} U^{\otimes t, t} \mu_{\text{Haar}}({\rm d}U)  \right
\Vert_{\infty} \leq \lambda,
\label{eq:TPE-def}\ee
with $U^{\otimes t, t} := U^{\otimes t} \otimes (U^{*})^{\otimes t}$.
We say $\nu$ is a $(N,\lambda,\infty)$-TPE if it is a
$(N,\lambda,t)$-TPE for all $t$.
\end{definition}
This definition is meant to generalize the spectral characterization
of expander graphs, and has the similar advantage that $|\supp(\nu)|$ can be
constant even for constant $\lambda$ and unbounded $N,t$
(c.f. \cite{HH09}).    Another advantage is that the TPE condition can
be naturally amplified.
For a distribution $\nu$, let $\nu^{*k}$ be the $k$-fold convolution of $\nu$, i.e. 
\be
\nu^{*k} = \int  \delta_{U_1...U_k} \nu({\rm d}U_1)...\nu({\rm d}U_k).
\ee
Then it follows immediately from \eq{TPE-def} and the fact that 
$  \int_{\mathbb{U}(N)} U^{\otimes t, t} \mu_{\text{Haar}}({\rm d}U)$
is a projector that
\be g(\nu^{*k},t) =
g(\nu,t)^k \label{eq:amplify-expander}. \ee

 \defref{TPE} can also be expressed in terms
of quantum operations.
Define $\ad_U[X] := UXU^\dag$. 
For a distribution $\nu$ on $\mathbb{U}(N)$ let
\begin{equation}
\Delta_{\nu, t}(\rho) := \int_{\mathbb{U}(N)} \ad_{U^{\ot t}}[\rho] \nu({\rm d}U)
= \int_{\mathbb{U}(N)}  U^{\otimes t} \rho
(U^{\cal y})^{\otimes t} \nu({\rm d}U).
\end{equation}

Define, for any $p\geq 1$, the superoperator norms
\be
\Vert {\cal T} \Vert_{p \rightarrow p}
 := \sup_{X \neq 0}  \frac{\Vert {\cal T}(X) \Vert_p} {\Vert X \Vert_p},
\ee
where $\|X\|_p := (\tr |X|^p)^{1/p}$ are the Schatten norms.
An alternate definition of the TPE condition is then:
\be g(\nu, t) = \| \Delta_{\nu, t} - \Delta_{\mu_{\text{Haar},
    t}}\|_{2\rar 2}. 
\label{eq:TPE-def-22}\ee

In many applications, however, it is often more natural to work
with measures such as  the trace distance.  For example, we would like to argue that
sampling $U$ from $\nu$ and using it $t$ times results in a state that
is $\eps$-close to one that would be obtained by sampling $U$ from the
Haar measure.  This will lead us to the notion of an {\em approximate
  unitary design} (also called an $\eps$-approximate $t$-design when
we want to emphasize the parameters)

Previous research has used several definitions of $\eps$-approximate
$t$-designs, such as replacing the $\|\cdot\|_\infty$ and $\lambda $ in \eq{TPE-def}
with $\|\cdot\|_1$ and $\eps$, or replacing the $2\rar 2$ norm in
\eq{TPE-def-22} with the diamond norm (defined below).  See
\cite{Low10} for a comparison of these, and other, ways of defining
approximate unitary designs.

Here we propose a stronger definition of approximate designs, which
was suggested to us by Andreas Winter.  First, if $\cN_1, \cN_2$ are
superoperators, then we say that $\cN_1\preceq \cN_2$ iff
$\cN_2-\cN_1$ is completely positive, or equivalently if 
\be
(\cN_1 \ot \id)(\Phi_N) \preceq (\cN_2 \ot \id)(\Phi_N),
\ee
where $N$ is the input dimension of $\cN_{1,2}$, $\preceq$ here
denotes the usual semidefinite ordering, and $\ket{\Phi_N} = N^{-1/2}
\sum_{i=1}^N \ket{i,i}$ is the standard maximally entangled state on
$N\times N$ dimensions.
 
\begin{definition}\label{def:design}
Let $\nu$ be a distribution on $\mathbb{U}(N)$. Then $\nu$ is an $\epsilon$-approximate unitary $t$-design if 
\be (1-\eps) \Delta_{\mu_{\text{Haar}, t}} \preceq
\Delta_{\nu, t} \preceq
 (1+\eps) \Delta_{\mu_{\text{Haar}, t}} 
\label{eq:design-def}\ee
or equivalently
\be (1-\eps) (\Delta_{\mu_{\text{Haar}, t}}\ot \id)(\Phi_N^{\ot t}) \preceq
(\Delta_{\nu, t}\ot\id)(\Phi_N^{\ot t}) \preceq
 (1+\eps) (\Delta_{\mu_{\text{Haar}, t}} \ot \id)(\Phi_N^{\ot t})
\label{eq:design-def-J}\ee
For brevity, let $G(\nu,t)$ denote the smallest $\eps$ for which
\eq{design-def} holds.
\end{definition}

The advantage of \defref{design} is that for any state on $t$ systems that
is acted upon by a random $U^{\ot t}$ and then measured, the
probability of any measurement outcome will change by only a small
multiplicative factor whether $U$ is drawn from $\nu$ or the Haar
measure.  To relate our design definition to the distinguishability of quantum
operations, we first define the diamond norm~\cite{Kitaev:02a} of a superoperator ${\cal T}$ as follows
\be
\left \Vert  {\cal T}  \right \Vert_{\diamond} := \sup_{d} \Vert {\cal T} \otimes \id_d \Vert_{1 \rightarrow 1},
\ee
\begin{lem}\label{lem:design-defs}
If $\nu$ is an $\eps$-approximate unitary $t$-design, then 
$\| \Delta_{\mu_{\text{Haar}, t}} -\Delta_{\nu, t}\|_\diamond \leq
2\eps$.
Conversely, if $\| \Delta_{\mu_{\text{Haar}, t}} -\Delta_{\nu, t}\|_\diamond \leq
\eps$ then $\nu$ is an $\eps N^{2t}$-approximate $t$-design.
\end{lem}
The proof is in \prettyref{app:design-facts}. 

The reason we should expect \lemref{design-defs} to be true is that
all norms on finite-dimensional spaces are equivalent, and every
definition of an approximate design is based on some norm of
$\Delta_{\nu,t} - \Delta_{\mu_{\text{Haar}}, t}$.  In practice, the
norms we are interested in always differ by factors that are
polynomial in dimension, which here means $N^{O(t)}$.  See Lemma
2.2.14 of \cite{Low10} for many more examples of this phenomenon.

To prove our main result about circuits being unitary designs, we will
take the common path of first showing that they are TPEs and then
converting this result into a statement about being designs.  This
conversion again loses a dimensional factor.

\begin{lem}\label{lem:design-expander}
Let $\nu$ be a distribution on $\bbU(N)$. Then
\be \frac{g(\nu,t)}{2N^{t/2}}  \leq G(\nu,t) \leq N^{2t} g(\nu, t).\ee
\end{lem}

This lemma is proved in \prettyref{app:design-facts}.

%Given a quantum $t$-copy tensor-product expander $(N, \lambda, t)$ one can always obtain an $\epsilon$-approximate unitary $t$-design with $\epsilon = ...$ by interation \cite{HL08}.

\section{Main Results}\label{sec:main}

\subsection{Haar Uniform Gates}
 We consider the following two models of random quantum circuits, defined as random walks on $\mathbb{U}(d^{n})$ for an integer $d$:
\begin{itemize}
%\item \textit{uniform random circuit}: in each step two indices $i \neq j$ are chosen uniformly at random from $[n]$ and a two-qudit unitary gate $U_{i, j}$ drawn from the Haar measure on $\mathbb{U}(d^{2})$ is applied to qubits $i$ and $j$.
\item \textit{Local random circuit}: In each step of the walk an index
  $i$ is chosen uniformly at random from $[n-1]$ and a unitary $U_{i,
    i+1}$ drawn from the Haar measure on $\mathbb{U}(d^{2})$ is
  applied to the two neighbouring qudits $i$ and $i+1$. 
% We consider qudit $n+1$ to be the same as qudit 1, so that qudits $n$ and 1 are neighbors.

\item \textit{Parallel local random circuit}: In each step either the
  unitary $U_{1,2} \otimes U_{3, 4} \otimes ... \otimes U_{n-1, n}$ or
  the unitary $U_{2, 3} \otimes ... \otimes U_{n-2, n-1}$ is applied
  (each with probability $1/2$), with $U_{j, j+1}$ independent
  unitaries drawn from the Haar measure on $\mathbb{U}(d^{2})$.  (This
  assumes $n$ is even.)
%in each step of the walk a matching $\{ (i_1, j_1), ..., (i_{n/2}, j_{n/2}) \}$ in $[n]$ is chosen uniformly at random and unitaries $U_{i_k, j_k}$ (with $k \in [n/2]$) drawn independently from the Haar measure on $\mathbb{U}(d^{2})$ are applied to the $n$ qubits.
\end{itemize}
We note the local random circuit model was considered previously in Refs. \cite{HP07} and \cite{BH10}, while a related model of parallel local random circuits, using a different set of quantum gates, was considered in Ref. \cite{EWS03}.

Denote the distribution over one step of a  local random circuit by
$\nu_{{\rm LR},d,n}$ and over one step of a parallel local random circuits by
$\nu_{{\rm PLR},d,n}$.   The distributions over circuits of length $k$
can be written as $\left( \nu_{{\rm LR},d,n} \right)^{*k}$ or $\left( \nu_{{\rm
    PLR},d,n} \right)^{*k}$, respectively. The main result of this paper is the following:
\begin{thm}\label{thm:main-TPE}\mbox{}
\benum
\item $g(\nu_{{\rm LR},d,n},t) \leq 1 - ( 42500 n
  \lceil\log_d(4t)\rceil^2 d^2 t^5 t^{3.1/\log(d)})^{-1}$.
%t^{2.5 \log(d)d^{-2} + 2.5\log(d)^{-1}})^{-1}$. 
\item $g(\nu_{{\rm PLR},d,n},t) \leq 1 - (553000 \lceil
  \log_d(4t)\rceil^2 d^2 t^5 t^{3.1/\log(d)})^{-1}$. 
  % t^{2.5 \log(d)d^{-2} + 2.5\log(d)^{-1}} )^{-1}$.
\eenum
\end{thm}

A direct consequence of this theorem, \eq{amplify-expander} and 
Lemma~\ref{lem:design-expander} is the following corollary about
forming $\eps$-approximate $t$-designs.
\begin{cor}\label{cor:main-design} \mbox{}
\benum 
\item
Local random circuits of length $42500  n
  \lceil\log_d(4t)\rceil^2 d^2 t^5 t^{3.1/\log(d)}
(2nt\log(d) + \log(1/\eps))$
form $\eps$-approximate $t$-designs. 
\item
Parallel local random circuits of length
$523000  \lceil
  \log_d(4t)\rceil^2 d^2 t^5 t^{3.1/\log(d)}
 (2nt\log(d) + \log(1/\eps))$ form
$\eps$-approximate $t$-designs. 
\eenum
\end{cor}

\vspace{0.3 cm}

\subsection{Other Universal Sets of Gates} Consider a set of
gates $G := \{ g_{i} \}_{i=1}^m$ with each $g_i \in
\mathbb{U}(d^2)$. We say $G$ is universal if the group generated by it
is dense in $\mathbb{U}(d^2)$, i.e. for every $g \in \mathbb{U}(d^2)$
and for every $\varepsilon > 0$ we can find a sequence $(i_1, ...,
i_L) \in [m]^L$ such that $\Vert g - g_{i_1}...g_{i_L} \Vert \leq
\varepsilon$. We say that the set $G$ contains inverses if $g^{-1} \in
G$ whenever $g \in G$. 

We can now consider random walks associated to an universal set of gates $G = \{ g_{i} \}_{i=1}^m$:
\begin{itemize}
%\item \textit{uniform random circuit}: in each step two indices $i \neq j$ are chosen uniformly at random from $[n]$ and a two-qudit unitary gate $U_{i, j}$ drawn from the Haar measure on $\mathbb{U}(d^{2})$ is applied to qubits $i$ and $j$.
\item \textit{$G$-local random circuit}: In each step of the walk two indices
  $i, k$ are chosen uniformly at random from $[n-1]$ and $[m]$, respectively, and the unitary $g_{k}$ is
  applied to the two neighboring qudits $i$ and $i+1$.  
  %We consider
  %qudit $n+1$ to be the same as qudit 1, so that qudits $n$ and 1 are neighbors.

\item \textit{$G$-parallel local random circuit}: In each step either the unitary $U_{1,2} \otimes U_{3, 4} \otimes ... \otimes U_{n-1, n}$ or the unitary $U_{2, 3} \otimes ... \otimes U_{n-2, n-1}$ is applied (each with probability $1/2$), with $U_{j, j+1}$ independent unitaries drawn uniformly from $G$.
%in each step of the walk a matching $\{ (i_1, j_1), ..., (i_{n/2}, j_{n/2}) \}$ in $[n]$ is chosen uniformly at random and unitaries $U_{i_k, j_k}$ (with $k \in [n/2]$) drawn independently from the Haar measure on $\mathbb{U}(d^{2})$ are applied to the $n$ qubits.
\end{itemize}

Corollary \ref{cor:main-design} only considered the case of a Haar uniform set of gates in $U(d^2)$. A natural question is whether one can prove similar results for other universal set of gates. It turns out that combining Theorem \ref{thm:main-TPE} with the result of Ref. \cite{BG11} one can indeed do so, at least for a large class of gate sets:

\def\corunisetofgates{
Fix $d \geq 2$. Let $G=\{ g_{i} \}_{i=1}^m$ be a universal set of gates
containing inverses, with each $g_{i} \in S\mathbb{U}(d^2)$ composed
of algebraic entries. Then there exists $C=C(G)>0$ such that
\benum 
\item
$G$-local random circuits of length 
 $C  n
  \lceil\log_d(4t)\rceil^2 t^5 t^{3.1/\log(d)}
(nt\log(d) + \log(1/\eps))$.
\item
$G$-parallel local random circuits of length $C
\lceil\log_d(4t)\rceil^2 t^5 t^{3.1/\log(d)} 
 (nt\log(d) + \log(1/\eps))$ form $\eps$-approximate $t$-designs. 
\eenum
}

\begin{cor} \torestate{\label{cor:unisetofgates}
\corunisetofgates}
\end{cor}

\corref{unisetofgates} is proved in \prettyref{sec:proof-univ}. The main tool behind the
proof is the beautiful result of Bourgain and Gamburd \cite{BG11} establishing that any
finite universal set of gates in $SU(N)$ (for us $N=d^2$), containing inverses and with
elements composed of algebraic entries, is an infinite tensor-product expander with
nonzero gap. We note that the proof in Ref. \cite{BG11} does not give any estimate of the
dependency of the spectral gap on $N$. That is the reason why Corollary
\ref{cor:unisetofgates} also does not specify the dependency of the size of the circuit on
the local dimension $d$. (And of course, this gap can be arbitrarily small, e.g. if the
gates in $G$ are all very close to the identity.)

\vspace{0.3 cm}

\subsection{Optimality of Results}
 It is worth asking whether these results can be improved.  In
\thmref{main-TPE}, we suspect that the dependence on $t$ could be
improved, perhaps even to obtain a gap that is independent of $t$.
In other words, we cannot rule out the possibility that random quantum
circuits form a $(d^n,\lambda,\infty)$-TPE for
$\lambda<1-1/\poly(n)$.  Indeed, taking $\nu$
to be uniform over even a constant number of Haar-random unitaries
from $\bbU(N)$ yields a $(N,\lambda,t)$-TPE for (according to
\cite{HH09}) $\lambda$ constant and $t$ as large as $N^{1/6-o(1)}$, 
and  (according to \cite{BG11}) $\lambda<1$ and $t=\infty$ (but with uncontrolled
$N$-dependence, and over a measure not quite the same as Haar). On the other hand, we can easily see that the
$n$-dependence of part 1 of \thmref{main-TPE} cannot be improved, and
the bound in part 2 is already independent of $n$. Even when
$t=1$, we can consider the action of a random circuit on a state whose
first qudit is in a pure state and the remaining qubits are maximally
mixed. Under one step of local random circuits, this state will
change by only $O(1/n)$.

What about designs?  Here, we can prove that neither the $t$ nor $n$
dependence can be improved by more than polynomial factors.

\def\converse{
Let $\nu$ be a distribution with support on circuits of size smaller
than $r$.  Suppose that $\nu$ is an $\eps$-approximate $t$-design
on $n$ qudits with $\eps\leq 1/4$ and $t\leq d^{n/2}$. Then
\be
r \geq \frac{nt}{5d^{4}\ln(nt)}. 
\ee
}

\begin{prop}\torestate{ \label{prop:converse}
\converse}
\end{prop}
We believe that the restriction on $\eps$ could be relaxed to
$\eps<1$, at the cost of some more algebra.  However, once $t\sim
d^n$, our lower bound must stop improving, since $O(d^{2n})$ two-qudit
gates suffice to implement any unitary~\cite{Shende:04a}, and in particular to achieve
the Haar measure.

\subsection{Classical analogues}

Random classical reversible circuits of size
$O(n^3t^2\log(n)\log(1/\eps))$ are known to 
generate $t$-designs~\cite{BH08} (with the caveat that 2-bit
reversible gates are not 
universal, so the base distribution needs to be over random 3-bit
gates).  Other work~\cite{KNR09,Kassabov05} implies that the number of
random bits in these constructions can be reduced to a nearly-optimal
$O(nt + \log(1/\eps))$.  Implicit in much of this work (e.g. see
\cite{HMMR05}) was an application similar to our Application I (below): namely,
producing permutations that could not be easily distinguished from a
uniformly random permutation.

Our techniques may be able to yield an alternate proof of
\cite{BH08}, possibly with sharper parameters.  However,
doing so would run into the difficulty that \lemref{propertiesHnt}
appears not to hold in the classical case.  To explain this, we
introduce some notation.  Let
$\Pi \vdash [2t]$ indicate that $\Pi$ is a partition of $[2t]$, let
$(x,y)\in \Pi$ indicate that $x,y$ are in the same block of $\Pi$, and
let $E_\Pi\subset [N]^{2t}$ be the set $\{(i_1,\ldots,i_{2t}) :
(x,y)\in\Pi \Rightarrow i_x=i_y \}$.  Then the
invariant subspace of $\{U^{\ot t,t} : U \in S(N)\}$ is spanned by the 
states 
\be
\ket{E_\Pi} := \frac{1}{\sqrt{|E_\Pi|}}
\sum_{(i_1,\ldots,i_{2t})\in E_\Pi} \ket{i_1,\ldots,i_{2t}}
\qquad \forall \Pi\vdash [2t]
\ee
However, here the analogy breaks down, because \eq{columnsum-first} is
known to fail for the states $\{\ket{E_\Pi}\}$. We are grateful to Kevin Zatloukal for
pointing this out to us.

\section{Applications} \label{sec:app}

Below we give four applications of \corref{main-design}.
\hspace{0.2 cm}

\noindent \textbf{Application 0: Large deviations.}  Matrix elements
of Haar-random unitaries obey concentration bounds very similar to
(indeed slightly stronger than) those of Gaussian random variables.
Specifically, if $U$ is a $d\times d$ Haar-random unitary matrix, then
for any unit vectors $\ket\alpha,\ket\beta$ and for $\gamma>0$ we have (see e.g. \cite{hayden2006aspects})

\be \Pr_U \L[|\bra\alpha U \ket\beta|^2 \geq \frac{\gamma}{d}\R]
\leq e^{-\gamma}.\label{eq:Haar-concentration}\ee
(If $U$ were instead a complex Gaussian matrix with the same first and
second moments, then the $\leq$ in \eq{Haar-concentration} would be
replaced by an $=$.)

If $U$ is instead an $\eps$-approximate $t$-design, then we can prove
the following analogue of \eq{Haar-concentration}.
\begin{lem}\label{lem:t-concentrate}
If $U$ is a $d\times d$ matrix distributed according to an
$\eps$-approximate $t$-design, then, for any unit vectors
$\ket\alpha,\ket\beta$ and any integer $\gamma\geq 1$, we have
\be \Pr_U \L[|\bra\alpha U \ket\beta|^2 \geq \frac{\gamma}{d}\R]
\leq (1+\eps)e^{-\min(t,\gamma)}.\label{eq:t-concentration}\ee
\end{lem}

Essentially it says that the exponential decay in $\gamma$ occurs up
until the point when $\gamma > t$, at which point we can only
guarantee that bad events happen with probability $\lesssim e^{-t}$.
Similar results were proved in \cite[Thm 1.2]{Low09} and \cite[Lemma
III.1]{BH10}.

\begin{proof}
Let $k = \min(t,\gamma)$.  Since $k\leq t$, we can say that $U$ is
distributed according to an $\eps$-approximate $k$-design.
\ban
\Pr_U\L[ |\bra\alpha U \ket\beta|^2 \geq \frac\gamma d\R]
& =\Pr_U\L[ |\bra\alpha U \ket\beta|^{2k} \geq \L(\frac\gamma d\R)^k\R]
\\ & \leq  \E |\bra\alpha  U \ket\beta|^{2k} \L(\frac d \gamma\R)^k
&\text{Markov's inequality}
\\ & \leq 
 (1+\eps)\frac{k!}{d(d+1)...(d+k-1)} \L(\frac d \gamma\R)^k
&\text{approximate design condition}\\
\\ & \leq (1+\eps)\frac{k!}{\gamma^k} 
\\ & \leq (1+\eps)\L(\frac{k}{e\gamma}\R)^k  & \text{Stirling}
\\ & \leq (1+\eps)e^{-k}
\ean
\end{proof}

\hspace{0.2 cm}

\noindent \textbf{Application I: Fooling small quantum circuits.} A
first application of \corref{main-design} is related to
pseudo-randomness properties of efficiently generated states and unitaries. A
folklore result in quantum information theory says that the
overwhelming majority of quantum states on $n$ qubits cannot be
distinguished (with a non-negligible bias) from the maximally mixed
state by any measurement which can be implemented by
subexponential-sized quantum circuits~\cite{BMW09,GFE09}. Thus, even
though such states 
are very well distinguishable from the maximally mixed state (since
they are pure), this is only the case by using unreasonable
measurements from a computational point of view. A drawback of this
result is that the states themselves require exponential-sized quantum
circuits even to be approximated (by applying the circuit to the
$\ket{0}^{\otimes n}$ state). For given an $n$-qubit state which can be prepared
by a circuit of $k$ gates, one can always distinguish it from the
maximally mixed state by a measurement implementable by $k+O(\log(1/\eps))$ gates: One
simply applies the conjugate unitary to the circuit which creates the
state (which is also a circuit of $k$ gates) and measures
$\log(1/\eps)$ qubits to see if they are each in the $\ket 0$ state or
not.

An interesting question in this respect posed in \cite{Aarblog} is the following: Can such a
form of data hiding (of whether one has a particular pure state or the
maximally mixed state) against bounded-sized quantum circuits be
realized efficiently? More concretely, can we find a state which can
be created by a circuit of size $s$, yet is indistinguishable from the
maximally mixed state by any measurement implementable by circuits of
size $r$, for $r$ sufficiently smaller than $s$? Using \corref{main-design} we can show that this is indeed the case. In fact, this
is a generic property of states that can be created by circuits of
size $s$:  
\def\corhiding{
Let $\left( \nu_{\text{LR}, d, n} \right)^{*s}$ be the distribution on $\mathbb{U}(d^{n})$ induced by $s$ steps of the local random quantum circuit model. Then for every $\delta \leq 1/20$,
\begin{eqnarray}
&&\Pr_{U \sim \left( \nu_{\text{LR}, d, n} \right)^{*s}  } \left(  \max_{M \hspace{0.05 cm} \in
    \hspace{0.05 cm} \text{size}(r)} \left | \bra{0^{n}}U^{\cal y} M U
    \ket{0^{n}} - \frac{\tr (M)}{d^{n}} \right | \geq \delta \right)
\leq \L(\frac{r}{\delta}\R)^{2rd^4}
\cdot 3 \left( \frac{560t}{d^{n}\delta^2}  \right)^{t/4},
\end{eqnarray}
with 
\be
 t = \L\lfloor \L(\frac{s}{1400 n^2\log(d)}\R)^{1/11}\R\rfloor,
\ee
 and where the maximization is realized over all
two-outcome POVMs $\{M, \id - M \}$ which can be implemented by
quantum circuits of size $r$. 

In particular, for fixed $d$ and $n$ sufficiently large, all but a $2^{-\Omega(n)}$-fraction of states generated by circuits of size $n^{k}$ (for $k \geq 10$) cannot be distinguished from the maximally mixed state with bias larger than $n^{- \Omega(1)}$ by any circuit of size $n^{\frac{k-9}{11}}\polylog^{-1}(n)$.
}
\begin{cor} \torestate{\label{cor:hiding}
\corhiding
}\end{cor}

A direct consequence of \corref{hiding} is connected to the problem of quantum circuit minimization. There we are given a quantum circuit consisting of $s$ gates and would like to determine the minimum number of gates which are needed to approximate the original circuit. We define
\be
C_{\epsilon}(U) := \min \left \{ k : \text{there exists} \hspace{0.1 cm} V  \hspace{0.1 cm} \text{with} \hspace{0.1 cm} k \hspace{0.1 cm} \text{gates s.t.} \hspace{0.1 cm} \Vert V - U\Vert_{\infty} \leq \epsilon  \right \}.
\ee
Then we can give a lower bound on $C_{\epsilon}(U)$ for a generic circuit $U$ using \corref{hiding} as follows
\begin{cor} \label{cor:circuit-LB}
All but a $2^{- \Omega(n)}$-fraction of quantum circuits $U$ of size $n^{k}$ (with respect to the measure $\left( \nu_{\text{LR}, d, n} \right)^{*n^k}$ induced by $n^{k}$ steps of the local random circuit model) satisfy $C_{\epsilon}(U) \geq n^{\frac{k-9}{11}}\polylog^{-1}(n)$ with $\epsilon := 1 - n^{- \Omega(1)}$.
\end{cor}

%It is an interesting open question whether \corref{circuit-LB} can be
%improved to say that most circuits of size $n^{k}$ cannot be
%approximated by circuits of size $n^{k - \epsilon}$, for a small
%$\epsilon > 0$.

A final result in this direction is that given a circuit in which Haar random unitaries are used a polynomial number of times, replacing them by random circuits only
incurs a small error.

\def\followingcircuitsbyhaar{
Let $C_U$ be a quantum circuit of size $r$ on $\leq r$ qudits that
makes use of a unitary oracle $U$ on $n\leq r$ qudits.  That is, each gate
in $C_U$ can either apply an arbitrary two-qudit gate to any pair of qudits, or
can apply $U$ to its first $n$ qudits.   Then 
\be \L\|\int_{\bbU(d^n)} \ad_{C_U} \nu_{{\rm LR},d,n}^{*s}({\rm d}U) - \int_{\bbU(d^n)}
\ad_{C_U} \mu_{\rm Haar}({\rm d}U)\R\|_\diamond\leq \epsilon,\ee
for any $\eps>0$ and $s \geq 42500 n r^{9.5} \log^2(r) (6nr d^2\log(d) + \log(1/\eps)) $. In other words,
random circuits cannot be distinguished from Haar-random unitaries by
significantly shorter circuits.
}

\begin{cor} \torestate{\label{cor:followingcircuitsbyhaar}
\followingcircuitsbyhaar}
\end{cor}

%The proofs of \corref{hiding} and \corref{followingcircuitsbyhaar} are
%in Section \ref{hiding-proof}.
%Proof: later.  {\em Note : this actually proves the stated claim in
%  the abstract.}

These corollaries are proved  in \appref{hiding-proof}.
\vspace{0.2 cm}

\noindent \textbf{Application II: Fast quantum equilibration.} A
second application of \corref{main-design} is related to dynamical
equilibration of subsystems of a time-evolving quantum
system. Understanding how a quantum system equilibrates despite
unitary global dynamics is a long-standing problem (see
e.g. \cite{GLTZ10}). Recently several new insights have been achieved
using ideas from quantum information theory \cite{LPSW09, CE10, GME11, Rie08, VZ11,MRA11, BCHHKM11, Cra11}. 

Here we outline two applications of our result to the problem of understanding equilibration in closed quantum dynamics. Consider the unitary time evolution of a system, initially in a fixed state, say all spins up $\ket{\uparrow}^{\otimes n}$. The total state at any particular time is pure and hence does not appear to equilibrate in any sense. However a long sequence of investigations, starting with von Neumann in 1929 \cite{Neu29}, has elucidated that the state does equilibrate if one imposes constraints on the kind of observations possible \cite{LPSW09, GLTZ06, TCFM12, BCH11}. For instance suppose that one only has access to measurements on a few of the particles. Then it turns out that the building up of entanglement in the quantum state leads to local equilibration of every small subset of particles, for almost all times \cite{LPSW09}. The limits of equilibration in closed quantum dynamics is an interesting problem. What is the largest class of observables for which equilibration holds? Our result on unitary designs allows us advance this question significantly. 

Recall the definition of the previous section of the circuit complexity of a measurement as the minimum size of any circuit of two-qubit gates that implements the measurement. Physical measurements (e.g. measurement of magnetization or heat capacity) of course have low complexity. An interpretation of Corollary  \ref{cor:hiding} is that in generic quantum dynamics given by random circuits (which model the case of generic evolutions under time-dependent Hamiltonians), the system equilibrates with respect to all measurements of low complexity. We note that this strong kind of equilibration has recently shown useful in understanding properties of black holes in the context of the AdS-CFT correspondence \cite{Sus14}. 

A second application is to strengthen a connection of \cite{MRA11} between the time of equlibration of small subsystems of a closed quantum system and the circuit complexity of the unitary which diagonalizes the Hamiltonian of the system.

Consider a quantum Hamiltonian on $n$ qudits (a Hermitian operator on $(\mathbb{C}^{d})^{\otimes n}$) which can be writen as 
\be
H = U D U^{\cal y}, 
\ee
with $D = \text{diag}(E_1, \ldots, E_{d^{n}})$ a diagonal matrix in
the computational basis formed by the eigenvalues of $H$ and $U$ a
unitary matrix. We divide the system into two subsystems $S$ and $E$,
where $S$ should be seen as a small subsystem and $E$ as a bath for
$S$. We consider an arbitrary initial state $\rho_{SE}(0)$ and its
time-evolved version $\rho_{SE}(t)=e^{-iHt }\rho_{SE}(0)  e^{iHt
}$. We are interested in the question of how quickly the subsystem state
$\rho_S(t) = \tr_E \left(\rho_{SE}(t) \right)$ reaches equilibrium (if
it equilibrates in the first place). The equilibrium state is denoted
by $\omega_{S}$ and is given by the reduced state of $S$ of the
time-averaged state
\begin{equation}
\omega_{SE} := \lim_{\tau \rightarrow \infty} \frac{1}{\tau}\int_{0}^{\tau} \rho_{SE}(t) dt = \sum_k P_k \left( \rho_{SE}(0) \right)P_k,
\end{equation}
with $P_k$ the eigenprojectors of the Hamiltonian $H$. 

In Refs. \cite{VZ11,MRA11,BCHHKM11} the square 2-norm average distance between $\rho_S(t)$ and $\omega_S$ was computed for a Hamiltonian with $U$ chosen from the Haar measure in $\mathbb{U}(d^{n})$:

\begin{lem}[Theorem 3 of \cite{BCHHKM11}]
\label{prop-HS}
\be \label{propHSaverage}
\int_{\mathbb{U}(d^{n})}\tr \left((\rho_S(t)-\omega_S)^2\right) \mu_{\text{Haar}}({\rm d}U)=c\left\{\frac{1}{d_S} \frac{|\eta|^2}{d_{ES}^2}+ \left( \frac{|\xi|^2}{d_{ES}^2}+
\frac{x}{d_{ES}^2} \right)^2+ O\biggl(\frac{1}{d_E}\biggr)\right\}
\ee
where $c$ is an absolute constant, $d_E,d_S$ are the dimensions of heat bath and the system, respectively, $d_{ES}=  d_E d_S=d^n$ is the dimension of the total system,  
\be
\xi = \sum_{k=1}^{d^n} e^{iE_kt}, \hspace{0.9 cm} \eta= \sum_{k=1}^{d^n} e^{i2E_kt},
\ee
and $x=\sum_k d_k^2$, with $d_k= \dim \left( P_k \right)$ being the dimension of eigenspace $P_{k}$.
\end{lem}

%Application of Levy's lemma and passing to trace norm yields:
%\begin{proposition}
%\label{prop-1-norm}
%With high probability according to Haar measure over $U$ it holds that 
%\be
%||\rho_S(t) - \omega_S||_1 \leq c' \left\{\frac{|\eta|}{d_{ES}}+
%\sqrt{d_S}\frac{|\xi|^4}{d_{ES}^4} + \sqrt{d_S} \frac{x}{d_{ES}^2}+O(\frac{d_S}{d_E}) + \sqrt{\frac{d_S}{d_{ES}^{1/3}}}\right\}
%\ee
%Here $c'$ is absolute constant, and the other notation is as in  Prop. \ref{prop-HS}. 
%\end{proposition}
It follows from Lemma \ref{prop-HS} that for a non-degenerate Hamiltonian the time average of the R.H.S. of Eq. (\ref{propHSaverage}) - over times of order of the inverse of the average energy gap - will be small (see Refs. \cite{VZ11,MRA11,BCHHKM11}). Thus a Hamiltonian whose basis is chosen according to the Haar measure and whose spectrum have on average large energy gaps (which is expected to be the case typically) will equilibrate rapidly. 

In Ref. \cite{MRA11}, Masanes, Roncaglia and Acin noted that the average computed in Lemma \ref{prop-HS} only involves polynomials in the entries of $U$ of degree $4$. Therefore one could consider the average over an $\epsilon$-approximate unitary $4$-design instead of the Haar measure and obtain the same result, up to an additive error of $\epsilon$. Also in Ref. \cite{MRA11} an interesting connection of fast equilibration and the circuit complexity of $U$ was put forward: It was argued that, assuming that random circuits of length $O(n^3)$ form an approximate unitary $4$-design, then the Hamiltonians of most circuits of such size enjoy fast equilibration of small subsystems. Conversely, a simple argument shows that circuits with complexity less than linear cannot lead to quick equilibration. Therefore there appears to exist a connection between fast local equilibration and the circuit complexity of the unitary diagonalizing the Hamiltonian.

\corref{main-design} allows us to strengthen this connection as follows: 
\begin{cor} \label{cor:equilibration}
For every $\delta > 0$, all but a $\delta$-fraction of quantum circuits $U$ of size $n^{k}$ (with respect to the measure $\left( \nu_{\text{LR}, d, n} \right)^{* n^k}$ induced by $n^{k}$ steps of the local random circuit model) are such that, with $H = UDU^{\cal y}$,
\be
\tr \left((\rho_S(T)-\omega_S)^2\right) \leq \frac{1}{\delta} \left( \int_{\mathbb{U}(d^{n})}\tr \left((\rho_S(T)-\omega_S)^2\right) \mu_{\text{Haar}}({\rm d}U) + 2^{- \Omega(n)} \right),
\ee
and
\be
C_{\epsilon}(U) \geq n^{\frac{k-9}{11}}\polylog^{-1}(n),
\ee
with $\epsilon := 1 - n^{- \Omega(1)}$.
\end{cor}

Using Corollary~\ref{cor:unisetofgates} we can also obtain an analogous statement for any universal set of gates containing inverses with elements formed by algebraic entries. 

In words, our results allows us to confirm the expectation of Ref. \cite{MRA11}  that random circuits form an approximate unitary $4$-design and also show that most such circuits indeed have large circuit complexity. The latter is useful information because one could worry that most unitaries of size $n^k$ (according to some chosen distribution on the set of circuits) would have a much shorter circuit decomposition, invalidating the connection of the time of equilibration with the circuit complexity of the diagonalizing unitary of the Hamiltonian. The fact that random circuits are not just approximate $4$-designs, but even approximate $\poly(n)$-designs is what allows us to prove \corref{circuit-LB} and show that indeed there is no such considerably shorter decomposition in general.

Although \corref{equilibration} makes clearer the connection of fast
subsystem equilibration and the complexity of diagonalizing the
Hamitonian, it is still is not the kind of statement one would hope
for. Indeed, to establish the connection in full one would like to
show that for most circuits for which $C_{\epsilon}(U)$ is large
enough (say, in the range $n^{k_1} < C_\eps(U) < n^{k_2}$ for all
sufficiently large $k_1 \ll k_2 = k_2(k_1)$), equilibration is fast. Here we can merely prove that most circuits $U$ of sufficiently large size are such that the corresponding Hamiltonian enjoys fast equilibration of small subsystems \textit{and} $C_{\epsilon}(U)$ is big. 

Another version of the claim that we would like to establish concerns
the {\em incompressibility}  of random circuits.  A strong version of
this conjecture would be that any $\eps$-covering of the set of
$t$-gate random circuits has cardinality $\geq (1/\eps)^{\Omega(t
  d^2)}$.  See \propref{converse} and \lemref{eps-net-circuits} for
some much weaker claims in this direction.  

One difficulty in establishing such a conjecture is that the exact
Hausdorff dimension of $r$-gate random circuits will depend on the
gauge freedom determined by their overlaps.  For example, an element
of $SU(4)$ has 15 real degrees of freedom, but three-qubit circuits of
the form $U_{12}U_{23}$ have $15+15-3$ degrees of freedom,
corresponding to the fact that the transformation
$U_{12},U_{23}\mapsto U_{12}V_2, V_2^\dag U_{23}$ leaves
$U_{12}U_{23}$ unchanged for any $V_2\in SU(2)$.

\hspace{0.2 cm}

\noindent \textbf{Application III: Generation of Topological Order.}
Topological order is a concept from condensed matter physics used to
describe phases of matter that cannot be described by the Landau local
order paradigm \cite{Wen90}. Roughly speaking topological order
corresponds to patterns of long-range entanglement in ground states of many-body Hamiltonians. The intrinsic stability of topologically ordered systems against local perturbations also make them attractive candidates for constructing robust quantum memories or even topological quantum computation \cite{Kit03, NSSFS08}.

In recent years it has emerged that it is fruitful to consider topological order as a property of quantum states, instead of quantum Hamiltonians (see e.g. \cite{BHF06, CCW10, Has11}). There are two approaches to define topologically ordered states. The first is to say that a state has topological quantum order (TQO) if it cannot be approximated by any state that can be generated by applying a local circuit of small depth to a product state. Thus the state contains multiparticle entanglement that cannot be created merely by local interactions. In more detail, an $n$-qubit state $\ket{\psi}$ defined on a lattice has $(R, \varepsilon)$ topological quantum order if for any parallel local circuit $U$ of depth $R$, $\Vert U \ket{0}^{\otimes n} - \ket{\psi} \Vert \geq \varepsilon$ \cite{BHF06, CCW10, Has11}.

The second approach is to say that a quantum state $\ket{\psi_0}$ defined on a lattice exhibits TQO if there is another state $\ket{\psi_1}$ orthogonal to it such that for all local observables $O_{\text{loc}}$, $\bra{\psi_0} O_{\text{loc}} \ket{\psi_0} \approx \bra{\psi_1} O_{\text{loc}} \ket{\psi_1}$ and $\bra{\psi_1} O_{\text{loc}} \ket{\psi_0} \approx 0$ \cite{BHF06, CCW10, Has11}. Thus one cannot distinguish the two states, or even any superposition of them, by local measurements. Quantitatively we say two orthogonal states $\ket{\psi_0}$, $\ket{\psi_1}$ defined on a finite dimensional lattice have $(l, \varepsilon)$-TQO if for any observable $O_{\text{loc}}$, with $\Vert O_{\text{loc}} \Vert \leq 1$, supported on a set of diameter less than $l$, we have $|\bra{\psi_0} O_{\text{local}} \ket{\psi_0} - \bra{\psi_1} O_{\text{loc}} \ket{\psi_1}| \leq 2\varepsilon$ and $|\bra{\psi_1} O_{\text{loc}} \ket{\psi_0}| \leq \varepsilon$. As shown in Ref. \cite{BHF06, Has11}, if a state is $(l, \varepsilon)$ topologically ordered according to the second definition, then it is also $(l/2, \varepsilon)$ topologically ordered according to the first definition. We remark in passing that topologically ordered states can also be understood as code states of any quantum error correcting code with large distance, and so the terminology ``topological order'' does not have to refer to any topological properties of the geometry of the qubits.

In Ref. \cite{BHF06} it was shown that in any fixed dimension $D$ a quantum evolution on $n$ qubits, in the form of a local Hamiltonian or a parallel local circuit, cannot generate topological quantum order in time (or depth in the case of a quantum circuit) less than $O(n^{1/D})$. The next corollary shows that in one dimension a generic evolution, chosen from the parallel local random circuit model, saturates this bound. Thus almost every local dynamics in 1D generates topological order at the fastest possible rate (according to the second and, hence, also first definition).

\def\topologicalorder#1{
There exists a universal constant $C<10^9$ such that for any two
orthogonal $n$-qubit states $\ket{\psi_0}$ and $\ket{\psi_1}$ most
random unitaries $U$
chosen from the measure $\left(\nu_{\text{PLR}, d,
    n}\right)^{* C n}$, induced by $Cn$ steps of the parallel local
random circuit model, map $\ket{\psi_0},\ket{\psi_1}$ to states with
nearly identical marginals.  More precisely, with probability larger
than $1 - 2^{-n/8}$, for every region $X = \{ l_0, ..., l_{0} + l \}$
of size $l \leq n/4$: 
\begin{equation} \label{eq:tqo1-#1}
\left \Vert \tr_{\backslash X} \left( U \ket{\psi_0}\bra{\psi_0}U^{\cal y} \right) - \tau_{X}  \right \Vert_1,  \left \Vert \tr_{\backslash X} \left( U \ket{\psi_1}\bra{\psi_1}U^{\cal y} \right) - \tau_{X} \right \Vert_1 \leq 2^{-n/8}
\end{equation}
and
\begin{equation}\label{eq:tqo2-#1}
\left \Vert \tr_{\backslash X} \left( U \ket{\psi_0}\bra{\psi_1}U^{\cal y} \right) \right \Vert_1 \leq 2^{-n/8} ,
\end{equation}
with $\tau_X$ the maximally mixed state in $X$ and $\tr_{\backslash X}$ the partial trace with respect to all sites except the ones in region $X$. Thus the states $U\ket{\psi_0}$ and $U\ket{\psi_1}$ exhibit $(n/4, 2^{-n/8})$-TQO. 
}
\begin{cor} \torestate{\label{cor:topologicalorder}
\topologicalorder{first}}
\end{cor}

The proof is in \secref{topo-order}.  The proof only uses the
2-design property of random circuits, which had already been
established in previous work~\cite{ODP06, DOP07, HL09, DJ10, AB08,
  Zni08}, but not previously for linear-depth circuits in one dimension.

Corollary \ref{cor:topologicalorder} also shows that one dimensional parallel random circuits \textit{scramble} \cite{HP07, SS08, LSHOH11} -- making an initial localized bit of information inaccessible to an observer that only looks at sublinear sized regions --  in linear time, confirming the expectation of Refs. \cite{HP07, SS08}.

\section{Proof overview of the main result}\label{sec:proof-OV}

\subsection{Local random circuits}

The proof of part 1 of \thmref{main-TPE} consists of four steps, explained below. 

\noindent \textbf{1. Relating to Spectral Gap:} In the first step,
following the work of Brown and Viola \cite{BV10} and Ref. \cite{BH10}
(see also the earlier work \cite{Zni08}), we rephrase the TPE
condition from \defref{TPE} in terms of the spectral gap of a local quantum Hamiltonian. A local
Hamiltonian on $n$ $D$-dimensional subsystems is a Hermitian matrix
$H$, acting on $(\mathbb{C}^{D})^{\otimes n}$, of the form $H = \sum_k
H_k$, where each $H_k$ acts non-trivially only on a constant number of
systems.  
The spectral gap of $H$, denoted by $\Delta(H)$, is normally given by
the absolute value of the difference of its two lowest distinct eigenvalues.  In our
setting the bottom $t!$ eigenvalues will be 0 and all other eigenvalues will be $>0$
(assuming that either that $U$ is drawn from the Haar measure or a set meeting the conditions of \corref{unisetofgates})
$\Delta(H)$ will equal the $t!+1^{\text{st}}$ eigenvalue.

Consider the following local Hamiltonian acting on $n$ subsystems,
each of dimension $D:=d^{2t}$:
\begin{equation} \label{Hntdef}
H_{n, t} := \sum_{i=1}^{n-1} h_{i, i+1}
\end{equation}
with local terms $h_{i, i+1} := I - P_{i, i+1}$ acting on subsystems $i, i+1$ and $P_{i, i+1}$ defined as
\begin{equation}\label{eq:Pii1-def}
P_{i, i+1} := \int_{\mathbb{U}(d^{2})} \left(U_{i, i+1}\right)^{\otimes t, t} \mu_{\text{Haar}}({\rm d}U).
\end{equation}
with $I$ the identity operator and $U^{\otimes t, t} := U^{\otimes t} \otimes (U^{*})^{\otimes t}$.

In section \ref{proofLemma1} we prove:
\def\lemHamGap{
\begin{equation}
g(\nu_{{\rm LR},n, d}, t) %= g(\mu_{n, d}, t)^{k} % = 
 =  1 - \frac{\Delta(H_{n, t})}{n}.
\end{equation}}
\begin{lem} \label{lem:Ham-gap}
\lemHamGap
\end{lem}
Lemma \ref{lem:Ham-gap} thus shows that in order to bound the
rate of convergence of the random walk associated to the random
quantum circuit (for its first $t$ moments), it sufficies to lower bound the spectral gap of $H_{n, t}$.

\vspace{0.2 cm}

\noindent \textbf{2. The Structure of $H_{n, t}$:} It turns out that $H_{n, t}$ has a few special properties which make the estimation of its spectral gap feasible. 

\def\lemPropertiesHnt#1{
For every $n, t > 0$ the following properties of $H_{n, t} = \sum_i (I - P_{i, i+1})$ hold:
\begin{enumerate}
\item  \cite{BV10} the minimum eigenvalue of $H_{n, t}$ is zero and the zero eigenspace is given by 
\begin{equation}
{\cal G}_{n, t} := \text{span} \left\{  \ket{\psi_{\pi, d}}^{\otimes
    n} : \hspace{0.2 cm} \ket{\psi_{\pi, d}} := (I \otimes V_{d}(\pi)) \ket{\Phi_d} \hspace{0.1 cm}, \hspace{0.1 cm}  \pi \in S_{t}  \right \},
\end{equation}
with $\ket{\Phi_{d^t}} := d^{-t/2}\sum_{k=1}^{d^{t}} \ket{k, k}$ the maximally entangled state on $(\mathbb{C}^{d})^{\otimes t} \otimes (\mathbb{C}^{d})^{\otimes t}$, $S_{t}$ the symmetric group of order $t$, and $V_{d}(\pi)$ the 
representation of the permutation $\pi \in \cS_t$ which acts on $(\mathbb{C}^{d})^{\otimes t}$ as
\begin{equation}
V_{d}(\pi) \ket{l_1} \otimes ... \otimes \ket{l_t} =
\ket{l_{\pi^{-1}(1)}} \otimes ... \otimes \ket{l_{\pi^{-1}(t)}};
\label{eq:Vd-def-#1}
\end{equation}
\item Let $G_{n, t}$ be the projector onto ${\cal G}_{n, t}$. If
  $t^2\leq d^n$, then 
\be \label{eq:columnsum-#1}
\sum_{\pi \in \cS_t} |\braket{\psi_{\sigma, d}}{\psi_{\pi, d}}|^{n} \leq
1 + \frac{t^{2}}{d^{n}},
\qquad\forall \sigma\in \cS_t
\ee
and 
\be \label{eq:normboundperm-#1}
\left \Vert  \sum_{\pi \in \cS_t}   \psi_{\pi, d}^{\otimes n} - G_{n, t} \right \Vert_{\infty} \leq \frac{t^{2}}{d^{n}}.
\ee
Here we use the convention that $\psi := \proj\psi$.
\end{enumerate}
}
\begin{lem} \label{lem:propertiesHnt}
\lemPropertiesHnt{first}
\end{lem}

In particular, the quasi-orthogonality property of the states
$\ket{\psi_{\pi}}$ given by Eqs. \eqs{columnsum-first}{normboundperm-first} will be necessary to derive a good lower bound
on the spectral gap of $H_{n, t}$.

%\begin{lem}
%For every integers $d > 1$,
%\be
%\sum_{\pi \in \cS_t} |\braket{\psi_{\sigma, d}}{\psi_{\pi, d}}| \leq 1 + \frac{t^{2}}{2d},
%\ee
%and
%\be
%\left \Vert  \sum_{\pi \in \cS_t}  \left( \ket{\psi_{\pi, d}}\bra{\psi_{\pi, d}} \right)^{\otimes n} - G_{n, t} \right \Vert_{\infty} \leq ...
%\ee
%with $G_{n, t}$ the projector onto ${\cal G}_{n, t}$.
%\end{lem}

\noindent \textbf{3. Lower Bounding the Spectral Gap:} With the
properties given by Lemma \ref{lem:propertiesHnt} we are in position to
lower bound $\Delta(H_{n, t})$. To this aim we use a result of
Nachtergaele \cite{Nac96}, originally proposed to lower bound the
spectral gap of frustration-free local Hamiltonians with a ground space
spanned by matrix-product states \cite{FNW92, PVWC07}. 
%Applied to the Hamiltonian $H_{n, t}$, it says that if it holds true that 
%\be \label{condnahtergaele}
%\Vert I_{A_1}\otimes G_{A_2B} \left( G_{A_1A_2} \otimes I_{B} - G_{A_1A_2B}  \right)   \Vert_{\infty} \leq \frac{1}{2l^{-1/2}}
%\sup \left \{ |\<\psi|\phi\>| \hspace{0.1 cm} : \hspace{0.1 cm}  \ket{\psi} \in \hcal_{A_1}\ot\gcal_{A_2B}, \hspace{0.1 cm} \ket{\phi} \in  \left( \gcal_{A_1A_2}\ot \hcal_B\right) \cap \gcal_{A_1A_2B}^{\perp}    \right \}  \leq \ep_l,
%\ee
%with $A_1:=[1,l/2]$, $A_2:=[l/2+1,l]$, $B:=n+1$, and $G_{[p, q]}$ the projector onto the eigenvalue zero eigenspace of $H_{q - p, t}$, then 
%\be
%\Delta(H_{n, t}) \geq \frac{\Delta(H_{l, t})}{2l}.
%\ee 
Using Nachtergaele's result in combination with Lemma
\ref{lem:propertiesHnt} we show in \secref{proof-shorten} the
following:
\def\lemNachtergaele#1{For every integers $n, t$ with  $n \geq \lceil   2.5 \log_d(4t) \rceil$,
\begin{equation}
\Delta(H_{n, t}) \geq \frac{\Delta(H_{\lceil 2.5 \log_d(4t) \rceil,
    t})}{ 4\lceil 2.5 \log_d(4t)\rceil}.
\label{eq:shorten-#1}
\end{equation}}

\begin{lem} 
\torestate{\label{lem:Nachtergaele}\lemNachtergaele{first}}
\end{lem}

Lemmas \ref{lem:Ham-gap} and \ref{lem:Nachtergaele} directly show that for every $t$, local random quantum circuits of polynomial size are a $\varepsilon$-approximate unitary $t$-design for every fixed $t$. Note, however, that they do not give any information about the dependence of $t$ on the size of the circuit.

\vspace{0.2 cm}

\noindent \textbf{4. Bounding Convergence with Path Coupling:} The
last step in the proof consists in lower bounding $\Delta(H_{\lceil
  2.5 \log_d(4t) \rceil, t})$. We achieve this by using the 
connection of the random circuit model problem with the spectral gap
of $H_{n, t}$ in the \textit{reverse} direction: We upper bound the
convergence time of the random walk on $\mathbb{U}(d^{n})$ defined by
the local random circuit in order to lower bound the spectral gap of
$\Delta(H_{\lceil 2.5 \log_d(4t) \rceil, t})$. The point is
that now any bound on the convergence time is useful. Actually, in
light of Lemma \ref{lem:Nachtergaele}, it sufficies to prove an
\textit{exponentially small} bound on the convergence time in order to
obtain part 1 of \thmref{main-TPE}, and this is what we accomplish.  

We consider the convergence of the random walk in the \textit{Wasserstein distance}
between two probability measures $\nu_1$ and $\nu_2$ on $\mathbb{U}(r)$.  The Wasserstein
distance is defined in terms of a distance measure $d(U,V)$ as follows:
\begin{equation}  \label{WassersteinDef}
W_{d,1}(\nu_1, \nu_2) := \sup  \left \{  \int_{\mathbb{U}(r)} f(U) \nu_1({\rm d}U) -  \int_{\mathbb{U}(r)} f(U) \nu_2({\rm d}U)  \hspace{0.2 cm} : \hspace{0.2 cm} f : \mathbb{U}(r) \rightarrow \mathbb{R} \hspace{0.2 cm} \text{is 1-Lipschitz}  \right \},  
\end{equation}
where we say that $f$ is 1-Lipschitz with respect to metric $d$ if for every two unitaries
$U, V$, $|f(U) - f(V)| \leq d(U, V)$.  The subscript 1 in $W_{d,1}$ is for consistency
with later notation when we will define $W_{d,p}$ general $p$.  In this paper we will
consider two distance measures.  The Frobenius distance $d_{\text{Fro}}(U,V)$ is
$\Vert U-V \Vert_2 := \tr((U-V)^{\cal y}(U-V))^{1/2}$.  The Riemannian distance
$d_{\text{Rie}}(U,V)$ is the length of a geodesic between $U$ and $V$, where the curve is
constrained to remain within the unitary group, i.e.
\be
d_{\text{Rie}}(U,V) = \min \left\{ \int_0^1 \| \gamma'(t) \|_2 dt  : \gamma(0)=U, \gamma(1)=V,
\gamma(t)\in \mathbb{U}(r)\, \forall t\in [0,1]\right\}.\ee
For brevity we write $W_{\text{Fro}} := W_{d_{\text{Fro}},1}$ and $W_{\text{Rie}} := W_{d_{\text{Rie}},1}$.
In section \ref{proofLemma3} we prove 
\begin{lem} \label{lconvergenceWasserstein}
For every integers $k, n > 0$,
\begin{equation}
  W_{\text{Fro}}((\nu_{{\rm LR},n, d})^{* (n - 1)k}, \mu_{\text{Haar}})
  \leq  \left( 1 - \frac{1}{e^{n}(d^2+1)^{n-2}} \right)^{\frac{k}{n-1}} \pi d^{n/2}.
\end{equation}
\end{lem}
The proof of Lemma \ref{lconvergenceWasserstein} rests on Bubley and
Dyer's \textit{path coupling} method \cite{BD97} for bounding the mixing time of Markov chains. In particular, we use a version of path coupling for Markov chains on the unitary group recently obtained by Oliveira \cite{Oli07} \footnote{In fact the result of Ref. \cite{Oli07} is more general and extends the path coupling method to Markov chains on a Polish length space.}.

Finally, it remains to show how Lemma \ref{lconvergenceWasserstein} implies a lower bound on the spectral gap of $\Delta(H_{\lceil 2 \log(d)^{-1} \log(t) \rceil, t})$. This is the content of the following Lemma, proved in section \ref{proofLemma4},
\def\lemWasserGap{
For every $t,d\geq 1$ and every measure $\nu$ on $\bbU(d^n)$,
\begin{equation} g(\nu,t) 
%\left( 1 - \frac{\Delta(H_{n, t})}{n} \right)^{k} 
\leq 2t W_{\text{Fro}}(\nu, \mu_{\text{Haar}})
\end{equation}
}
\begin{lem} \label{lem:WasserGap}
\lemWasserGap
\end{lem}
Part 1 of \thmref{main-TPE} now follows from the previous lemmas. 

\vspace{0.4 cm}

\begin{proof} \textbf{(Part 1 of Theorem \ref{thm:main-TPE})}
Lemmas \ref{lem:Ham-gap}, \ref{lconvergenceWasserstein} and
\ref{lem:WasserGap}, along with \eq{amplify-expander} give that
for every $m, t, k$,  
\begin{equation}
1 - \frac{\Delta(H_{m, t})}{m} \leq (2t \pi d^{m/2})^{\frac{1}{k(m-1)}} \left( 1 - \frac{1}{e^{m}(d^2+1)^{m-2}} \right)^{\frac{1}{(m-1)^{2}}} .
\end{equation}
Taking the $k \rightarrow \infty$ limit we find,
\begin{equation}
\Delta(H_{m, t}) \geq m^{-1} e^{-m}(d^{2} + 1)^{-m}. 
\end{equation}

Then by Lemma \ref{lem:Nachtergaele}  and the previous equation, with
$m =  \lceil 2.5 \log_d(4t)\rceil $, we get that for every $n$, 
\begin{eqnarray} \label{eq:boundt4logt}
\Delta(H_{n, t})  &\geq& 
\frac{ t^{-\frac{2.5}{\log(d)} - 2.5 \frac{\log(d^2+1)}{\log(d)} }}
{12500 \lceil\log_d(4t)\rceil^2e(d^2+1)}
% &=&  ( 5 \log(4t)\log(d)^{-1})^{-2}      t^{-4 - 2 \log(d)/d^2 - 2/\log(d)} \nonumber \\
\end{eqnarray}
%This exponent can be bounded as
%$$\frac{2.5}{\log(d)} + \frac{2.5\log(d^2+1)}{\log(d)}
% = 5 + \frac{2.5(1 + \log(1+d^{-2})}{\log(d)} \leq 10$$ 
Since $2.5(1 + \log(1+d^{-2})) \leq 3.1$ and $125e(1+d^{-2})\leq 42500$ for $d\geq 2$, 
our result now follows from \lemref{Ham-gap}.
\end{proof}

\subsection{Parallel local random circuits}\label{sec:PLR}

To analyze parallel local random circuits and prove part 2 of \thmref{main-TPE}, we use part 1 of \thmref{main-TPE} and a recent tool for analysing quantum many-body Hamiltonians: the detectability lemma of Aharonov \textit{et al} \cite{AALV11}.

Define 
\begin{equation}\label{eq:Mnt-def}
M_{n, t} := \frac{1}{2} P_{1, 2}  P_{3, 4} \cdots P_{n-1, n} + \frac{1}{2} P_{2, 3} \cdots
P_{n-2, n-1},
\end{equation}
(with $P_{i,i+1}$ defined in \prettyref{eq:Pii1-def}) and let $\lambda_{t!+1}(M_{n, t})$ denote its $t!+1$'st largest eigenvalue. (We focus on
this eigenvalue since the top $t!$ eigenvalues are always 1.) In analogy with \lemref{Ham-gap} it holds that
\begin{equation} \label{anaoguelemma1}
\left \Vert \int_{\mathbb{U}(d^{n})} U^{\otimes t, t} \nu_{\text{PLR}, d, n}({\rm d}U) -
  \int_{\mathbb{U}(d^{n})} U^{\otimes t, t} \mu_{\text{all}}({\rm d}U)
\right \Vert_{\infty} = \lambda_{t!+1}(M_{n, t})
\end{equation}

Let $P_{\text{odd}} := P_{1, 2}  P_{3, 4} \cdots P_{n-1, n}$, $P_{\text{even}}
:= P_{2, 3} \cdots  P_{n-2, n-1}$ and $P_{\text{all}}$ be the projector onto the
intersection of $P_{\text{odd}}$ and $P_{\text{even}}$.  We use the notation
$P_{\text{all}}$ because  $P_{\text{all}} = \int_{\mathbb{U}(d^{n})} U^{\otimes t, t} \mu_{\text{all}}({\rm d}U)$.
Then 
\begin{lem} \label{auxlemmatheo2}
\begin{equation}
\lambda_{t!+1}(M_{n, t}) \leq \frac{1}{2} + \frac{1}{2}\Vert  P_{\text{odd}} P_{\text{even}} - P_{\text{all}} \Vert_{\infty}.
\end{equation}
\end{lem}
\begin{proof}
 We make use of the following result of \cite{DFSS05} (Proposition 2.4): Given two projectors $Q$ and $R$, $\Vert Q + R \Vert \leq 1 + \Vert Q R \Vert$. Let $P_{\text{all}}$ be the projector onto the intersection of $P_{\text{odd}}$ and $P_{\text{even}}$. Applying the previous inequality with $Q = P_{\text{odd}} - P_{\text{all}}$ and $R = P_{\text{even}} - P_{\text{all}}$,
\begin{equation}
\Vert P_{\text{odd}} + P_{\text{even}} - 2P_{\text{all}} \Vert_{\infty} \leq 1 + \Vert  P_{\text{odd}} P_{\text{even}} - P_{\text{all}} \Vert_{\infty},
\end{equation}
and so
\begin{equation}
\lambda_{t!+1}(M_{n, t}) = \Vert M_{n, t} - P_{\text{all}} \Vert_{\infty} \leq \frac{1}{2} + \frac{1}{2}\Vert  P_{\text{odd}} P_{\text{even}} - P_{\text{all}} \Vert_{\infty}.
\end{equation}
\end{proof}

Now using the detectability lemma \cite{AALV11} we can show:
\begin{lem} \label{lem:detectabilityconsequence}
\begin{equation}\label{eq:detect-bound}
\lambda_{t!+1} \left(  M_{n, t} \right) \leq \frac{1}{2} + \frac{1}{2}\left(1 + \frac{\Delta(H_{n, t})}{2} \right)^{- 1/3}.
\end{equation}
\end{lem}
\begin{proof}
Since $H_{n, t}$ is a frustration-free Hamiltonian with projective local terms we can apply the detectability lemma, which is the following bound
\begin{equation}
\Vert  P_{\text{odd}} P_{\text{even}} - P_{\text{all}} \Vert_{\infty} \leq \left(1 + \frac{\Delta(H_{n, t})}{2} \right)^{- 1/3}.
\end{equation}
The statement of the lemma thus follows from Lemma \ref{auxlemmatheo2}.
\end{proof}

To evaluate the bound in \prettyref{eq:detect-bound}, note that the RHS is
$\leq 1 - \frac{1}{12}\Delta(H_{n,t}) + O(\Delta(H_{n,t})^2)$.   From part 1 of
\thmref{main-TPE} we know that $\Delta(H_{n,t}) \geq \delta$ for some $\delta\leq
1/42500$.  Thus the gap in the PLR model is at least $\frac{n}{13}$ times the bound in the
LR model.  This completes the proof of part 2 of \thmref{main-TPE}.

\section{Proof of Lemmas for \thmref{main-TPE}}\label{sec:proof-details}

\subsection{Proof of Lemma \ref{lem:Ham-gap}} \label{proofLemma1}
We start proving Lemma \ref{lem:Ham-gap}, which is restated below 
for the convenience of the reader.

\begin{replem}{lem:Ham-gap}
\lemHamGap
\end{replem}

\begin{proof} 
The lemma follows from 
\begin{eqnarray}
g(\nu_{\text{LR}, n, d}, t) &=& \lambda_{t!+1} \left(  \int_{\mathbb{U}(d^{n})} U^{\otimes t, t} \nu_{\text{LR}, n, d}({\rm d}U)  \right) \nonumber \\ &=& \lambda_{t!+1} \left(  \frac{1}{n} \sum_{i} P_{i, i+1}  \right) \nonumber \\ &=& \lambda_{t!+1} \left( I - \frac{H_{n, t}}{n}  \right) = 1 - \frac{\Delta(H_{n, t})}{n},
 \end{eqnarray}
with $\lambda_{t!+1}(X)$ the $t!+1^{\text{st}}$ largest eigenvalue of $X$.
\end{proof}

\subsection{Properties of $H_{n, t}$} \label{propertiesH}

We now prove \lemref{propertiesHnt}.
\begin{replem}{lem:propertiesHnt}
\lemPropertiesHnt{second}
\end{replem}

\begin{proof}
\textit{Item 1.}  Since each $P_{i,i+1} \leq I$, we have that the smallest
eigenvalue of $H$ is $\geq 0$.  Let us now determine the ground
space.
\ba H_{n,t}\ket\varphi =0
& \Leftrightarrow \frac{1}{n}\sum_{i=1}^n P_{i,i+1}\ket\varphi =
\ket\varphi \\
& \Leftrightarrow \forall i\in [n],\forall U\in U(d^2), 
(I_d^{\ot i-1} \ot U_{i, i+1} \ot I_d^{\ot n-i-1})^{\ot t,t}  \ket\varphi = \ket\varphi
\label{eq:bc-triangle} \\ 
& \Leftrightarrow \forall U\in \bbU(d^n), U^{\ot t,t}
\ket\varphi = \ket\varphi
\label{eq:bc-universal}
\ea
Here % \eq{bc-positive} is because each $h_{i,i+1}\geq 0$ and
Eq. \eq{bc-universal} is because nearest-neighbor unitaries generate the
set of all unitaries~\cite{Bar+95}.  To justify \eq{bc-triangle},
observe that 
\[ \text{Re} \bra\varphi \E_{i\in [n]} P_{i,i+1} \ket\varphi = 
\text{Re} \E_{i\sim [n]} \E_{U_{i,i+1}} \bra\varphi U_{i,i+1}^{t, t} \ket\varphi
\leq 1 \]
with equality if and only if $U_{i,i+1}^{\otimes t, t} \ket\varphi = \ket\varphi$ for
all but a measure-zero subset of the $(i, U_{i,i+1})$ pairs. And by continuity, 
we can assume this subset is empty.  
%(See also Lemma 3.2 of \cite{HL09} for a similar argument.)

We can without loss of generality write $\ket{\varphi} = (I_d^{\ot nt} \ot
M)\ket{\Phi_{d^{nt}}}$ for some matrix $M$.  In terms of $M$,
Eq. \eq{bc-universal} implies that $\ket\varphi$ 
is a ground state of $H_{n,t}$ if and only if $M$ commutes with
$U^{\ot t}$ for all $U\in U(d^n)$.  It is well-known (see \cite{GW98},
or \cite{matthias} for a quantum information perspective) that the set of
such $M$ is precisely given by the span of the $V_d(\pi)$ for $\pi\in
\cS_t$. \footnote{Note that the form of the eigenspace of $H_{n, t}$
  follows directly from the fact that random circuits drawn from a
  universal set of gates converge to the Haar measure, a fact that was
  first proven at least as early as in \cite{AK62}. A more direct proof of convergence can
  be also obtained by applying general sufficient conditions given by
  Theorem 3.3 of \cite{Szarek06} for Markov chains to converge to a
  unique invariant measure \cite{TS_priv}.} 
%The crucial condition is universality of $\mu_{n,d}$ 
%according to Definition 2 of Ref. \cite{}.

\textit{Item 2.} Eq. \eq{columnsum-second} follows from 
\begin{eqnarray} \label{lemma743744}
\sum_{\pi \in \cS_t} |\braket{\psi_{\sigma, d}}{\psi_{\pi, d}}|^{n} &=&  \frac{1}{d^{tn}}  \sum_{\pi \in \cS_t}\tr \left( V_{d^{n}}(\pi) V_{d^{n}}(\sigma)^{T} \right) \nonumber \\
&=&  \frac{1}{d^{tn}} \sum_{\pi \in \cS_t}  \tr \left( V_{d^{n}}(\pi \sigma^{-1}) \right) \nonumber \\
&=&  \frac{1}{d^{tn}}  \sum_{\pi \in \cS_t}  \tr \left( V_{d^{n}}(\pi) \right) \nonumber \\
&=& \frac{t!}{d^{tn}} \tr \left( P_{\sym, t, d^{n}}   \right),
\end{eqnarray}
with $P_{\sym, t, d^{n}}$ the projector onto the symmetric subspace of $\left( \mathbb{C}^{d^{n}} \right)^{\otimes t}$. The first equality follows from the definition of $\ket{\psi_{\pi, d}}$ and the relation $V_{d^{n}}(\pi) = \left( V_{d}(\pi)  \right)^{ \otimes n}$, the second and third from the fact that $\cS_t$ is a group and $V_{d^{n}}(\pi)$ a representation of $\pi$, and the last from the relation
\be
P_{\sym, t, d^{n}} = \frac{1}{t!} \sum_{\pi \in \cS_t}V_{d^{n}}(\pi) .
\ee
Using $\tr(P_{\sym, t, d^{n}}) = (d^{n} + t - 1) ... (d^{n} + 1)d^{n} / t!$ , Eq. (\ref{lemma743744}), and our assumption that $t^2\leq d^n$, we obtain
\begin{equation}
\sum_{\pi \in \cS_t} |\braket{\psi_{\sigma, d}}{\psi_{\pi, d}}|^{n} = \frac{(d^{n} + t - 1) ... (d^{n} + 1)d^{n} }{d^{tn}} \leq 1 + \frac{t^{2}}{d^{n}}.
\label{eq:overlap-sum}
\end{equation}

To prove Eq. \eq{normboundperm-second}, let $B := \sum_{\pi \in \cS_t} \ket{\pi}\bra{\psi_{\pi, d}}^{\otimes n}$, with $ \{ \ket{\pi} \}_{\pi \in \cS_t}$ an orthornomal set of vectors. We have 
\begin{equation}
\left \Vert BB^{\cal y} - \sum_{\pi \in \cS_t} \ket{\pi}\bra{\pi} \right \Vert_{\infty}
\leq \max_\sigma \sum_{\pi \neq \sigma}  |\braket{\psi_{\sigma, d}}{\psi_{\pi, d}}|^{n}
\leq  \frac{t^{2}}{d^{n}},
\end{equation}
where we used Eq. \eq{overlap-sum} and the fact that for $M$ a Hermitian matrix,
$\|M\|_\infty \leq \max_i \sum_j |M_{ij}|$. (The proof of this last claim is essentially
Perron-Frobenius: observe that if $M\ket\varphi = \lambda\ket\varphi$ and $i=\arg\max_i
|\braket{i}{\varphi}|$ then $|\lambda| |\braket{i}{\varphi}| = |\sum_j M_{i,j}
\braket{j}{\varphi}| \leq \sum_j |M_{i,j}|\cdot |\braket{i}{\varphi}|$.)
Since $BB^{\cal y}$ has the same eigenvalues as $B^{\cal y}B$ and 
\be
A_{n, t} := \sum_{\pi \in \cS_t}  \left( \ket{\psi_{\pi, d}}\bra{\psi_{\pi, d}} \right)^{\otimes n} = B^{\cal y}B,
\ee
we find $(1 - \frac{t^{2}}{d^{n}})G_{n, t} \leq A_{n, t} \leq (1 + \frac{t^{2}}{d^{n}}) G_{n, t}$, where we used that $G_{n, t}$ is the projector onto the support of $A_{n, t}$. Thus
\begin{eqnarray} 
\left \Vert  A_{n, t} - G_{n, t} \right \Vert_{\infty} \leq \frac{t^{2}}{d^{n}},
\end{eqnarray}
which is Eq. \eq{normboundperm-second}.
\end{proof}
The facts used here about $P_{\sym,t,d^n}$ are proved and discussed in
a quantum-information setting in \cite{Har-sym}.

%An important element in the proof will be the structure of the ground space ${\cal G}_{n, t}$. The following lemma shows that it can be approximated by 
%\begin{lem}
%For an integer $r > 1$,
%\be
%\sum_{\pi \neq \sigma} |\braket{\psi_{\sigma, t, r}}{\psi_{\pi, t, r}}|  \leq,
%\ee
%and
%\be
%\left \Vert  \sum_{\pi \in \cS_t}  (\ket{\psi_{\pi, t, r}}\bra{\psi_{\pi, t, r}})^{\otimes n}  \right \Vert_{\infty} \leq
%\ee
%\end{lem}

%\begin{proof}
%Let us start with the first inequality. We have
%\be
%\braket{\psi_{\sigma}}{\psi_{\pi}} = \frac{\tr(P^{T}_{t, r}(\sigma)P_{t, r}(\pi))}{r^{t}} = \frac{\tr(P_{\sigma^{-1}\pi, t, r}}{r^{t}}.
%\ee
%Thus
%\be \label{connectionwithdimPsym}
%\sum_{\pi \in \cS_t} |\braket{\psi_{\sigma, t, r}}{\psi_{\pi, t, r}}|  =  \sum_{\pi \in \cS_t} \frac{\tr(P_{t, r}(\pi)}{r^{t}} = \frac{t!}{r^{t}} \tr(P_{\text{sym}, t, r}),
%\ee
%with $P_{\text{sym}, t, r}$ the projector onto the symmetric subspace of $\left( \mathbb{C}^{r} \right)^{t}$. The last equality in Eq. (\ref{connectionwithdimPsym}) %follows from the relation
%\be
%P_{\text{sym}, t, r} = \frac{1}{t!} \sum_{\pi \in \cS_t} P_{\pi, t, r}
%\ee
%Then using $\tr(P_{\text{sym}, t, r}) = r(r+1)...(r+t-1)/t!$ we can prove 
%\begin{equation}
%\frac{t!}{r^{t}} \tr(P_{\text{sym}, t, r}) \leq 1 + \frac{t^{2}}{2r}. 
%\end{equation}
%from which Eq. () follows. 

%For Eq. ()...
%\end{proof}

\subsection{Proof of Lemma \ref{lem:Nachtergaele}} \label{sec:proof-shorten}

%In this section we prove Lemma \ref{lem:Nachtergaele}, which lower bounds the spectral gap of $H_{n, t}$ by the spectral gap of $H_{\lceil 6 t \log(t) \rceil, t}$, which manifestly independent of $n$ and larger than zero. The key ingredient in the proof is a result of Nachtergaele which gives a method for lower bounding the spectral gap of local frustration free Hamiltonians with a ground space spanned by matrix product states. Although Nachtergaele result can be used as a black-box to give a full proof of Lemma \ref{lem:Nachtergaele}, we choose to rederive it based on a more fundamental result also by Nachtegaele. We do so in order to have a self contained proof and also because we can prove our result by simpler means than used in [], in part due to the fact that $H_{n, t}$ has a ground space spanned by product states.

We start defining the necessary notation to state the result of \cite{Nac96} which we employ. We consider a chain 
of systems with local finite dimensional Hilbert space $\hcal$ labeled by natural numbers (excluding $0$). 
We consider a family of Hamiltonians 
\be
H_{[m,n]}=\sum_{i=m}^{n-1} h_{i, i+1}
\ee
acting on $\hcal^{\otimes (n - m + 1)}$, where $h_{i, i+1}$ are the nearest neighbor interaction terms, which are assumed to be projectors. In words, $H_{[m,n]}$ 
includes all the interactions terms for which both systems 
belong to the interval $[m,n]$. We also let the chain be translationally invariant, 
i.e. $h_{i, i+1}$ are the same for all $i$. 
We assume further that the minimum eigenvalue of $H_{[m, n]}$ is zero for all $m, n$ and denote by $\gcal_{[m,n]}$ the ground space of $H_{[m,n]}$, namely
\begin{equation}
\gcal_{[m,n]}=\{\ket{\psi} \in \hcal^{\otimes (n - m+1)} \hspace{0.2 cm}:  \hspace{0.2 cm} H_{[m,n]}\ket{\psi}=0\}.
\end{equation}
Finally let $G_{[m, n]}$ be the projector onto $\gcal_{[m,n]}$.
%We assume that $H_{[1,l]}$ has a nontrivial ground space.
%, and has a gap $\gamma_l$. 
%Note that due to translation invariance this is true for the Hamiltonian of any $l$ consecutive systems. 
%Finally, for any two integers $n$ and $l$ we define three subsystems:
%$A_1=[1,n-l+1]$, $A_2=[n-l,n]$, $B=n+1$. 

%Under the above notation and assumptions the following holds
\begin{lem}[Nachtergaele, Theorem 3 of \cite{Nac96}] 
\label{prop:gap-N}
Suppose there exist positive integers $l$ and $n_l$, and a real number $\ep_l \leq 1/\sqrt{l}$ such that for all $n_l \leq m \leq n$,
\be
\Vert I_{A_1}\otimes G_{A_2B} \left( G_{A_1A_2} \otimes I_{B} - G_{A_1A_2B}  \right)   \Vert_{\infty} \leq \ep_l
%\sup \left \{ |\<\psi|\phi\>| \hspace{0.1 cm} : \hspace{0.1 cm}  \ket{\psi} \in \hcal_{A_1}\ot\gcal_{A_2B}, \hspace{0.1 cm} \ket{\phi} \in  \left( \gcal_{A_1A_2}\ot \hcal_B\right) \cap \gcal_{A_1A_2B}^{\perp}    \right \}  \leq \ep_l,
\ee
with $A_1:=[1,m- l - 1]$, $A_2:=[m-l,m - 1]$, $B:=m$. 
%where the supremum is taken over all vectors $\ket{\psi}$ belonging to $\hcal_{A_1}\ot\gcal_{A_2B}$ and vectors $\ket{\phi}$ belonging to  $G_{A_1A_2}\ot \hcal_B$ and orthogonal to $\gcal_{A_1A_2B}$ \footnote{Recall that $A_1,A_2$ and $B$ depend both on $n$ and on $l$.}. 
Then
\be
\Delta(H_{[1,n]}) \geq \Delta(H_{[1,l]})  \frac{(1-\ep_l \sqrt{l})^{2}}{l-1}.
\ee 
\end{lem}
\vspace{0.2 cm}
We can now prove \lemref{Nachtergaele}, restated below:

\begin{replem}{lem:Nachtergaele}
\lemNachtergaele{second}
\end{replem}

\begin{proof}
We apply Lemma \ref{prop:gap-N} with $n_l = 2l$ and $\epsilon_l = 1/(2\sqrt{l})$. Then we must show that for all $m$ in the range $2l \leq m \leq n$, 
\be
\Vert I_{A_1}\otimes G_{A_2B} \left( G_{A_1A_2} \otimes I_{B} - G_{A_1A_2B}  \right)   \Vert_{\infty} \leq \frac{1}{2 \sqrt{l}},
\ee
with $A_1 = [1, m - l - 1]$, $A_2 = [m - l, m - 1]$ and $B = m$. Let
\be
X_{k} := \sum_{\pi \in \cS_t} (\ket{\psi_{\pi, d}}\bra{\psi_{\pi, d}})^{\otimes k}.
\ee
By Eq. \eq{normboundperm-second} of Lemma \ref{lem:propertiesHnt} we have 
\be
\Vert G_{[1, ..., k]} - X_{k} \Vert_{\infty} \leq \frac{t^{2}}{d^{k}}. 
\ee
Then
\begin{eqnarray}
M &:=& \Vert I_{A_1}\otimes G_{A_2B} \left( G_{A_1A_2} \otimes I_{B} -
       G_{A_1A_2B}  \right)   \Vert_{\infty} \nonumber \\  
&\leq& \left \Vert  I_{A_1} \otimes  X_{l+1} \left( X_{m-1}  \otimes
       I_{B} - X_{m} \right) \right \Vert_{\infty}  +
       \frac{5t^{2}}{d^{l}}  \nonumber \\  
&=& \left \Vert \sum_{\pi \in \cS_t} (\ket{\psi_{\pi,
    d}}\bra{\psi_{\pi, d}})^{\otimes (m - l - 1)} \otimes Y_{\pi}
    \right \Vert_{\infty} + \frac{5t^{2}}{d^{l}}. 
\end{eqnarray}
with
\begin{eqnarray}
Y_{\pi} &:=&  \sum_{\sigma \neq \pi}  \left(\ket{\psi_{\sigma, d}} \bra{\psi_{\sigma, d}})^{\otimes l} (\ket{\psi_{\pi, d}}\bra{\psi_{\pi, d}})^{\otimes l} \right) \otimes \left( \ket{\psi_{\sigma, d}} \bra{\psi_{\sigma, d}}(I_B - \ket{\psi_{\pi, d}} \bra{\psi_{\pi, d}})  \right) \nonumber \\
&=& \sum_{\sigma \neq \pi} (\braket{\psi_{\sigma, d}}{\psi_{\pi, d}})^{l}(\ket{\psi_{\pi, d}}\bra{\psi_{\sigma, d}})^{\otimes l} \otimes  \left( \ket{\psi_{\sigma, d}} \bra{\psi_{\sigma, d}}(I_B - \ket{\psi_{\pi, d}} \bra{\psi_{\pi, d}})  \right).
\end{eqnarray}
In the remainder of the proof we show
\be \label{blockdiagonalform}
\left \Vert \sum_{\pi \in \cS_t} (\ket{\psi_{\pi, d}}\bra{\psi_{\pi, d}})^{\otimes (m - l - 1)} \otimes Y_{\pi} \right \Vert_{\infty} \leq \left( 1 + \frac{t^{2}}{d^{m - l - 1}} \right)\max_{\pi} \Vert Y_{\pi} \Vert_{\infty}.
\ee
Then,
\begin{eqnarray}
M &\leq& \left( 1 + \frac{t^{2}}{d^{m - l - 1}} \right)\max_{\pi}
         \Vert Y_{\pi} \Vert_{\infty}  + \frac{5t^{2}}{d^{l}}
         \nonumber \\ 
&\leq& \left( 1 + \frac{t^{2}}{d^{m - l - 1}} \right) \max_{\pi}
       \sum_{\sigma \neq \pi} |\braket{\psi_{\sigma, d}}{\psi_{\pi,
       d}}|^{2l} + \frac{5t^{2}}{d^{l}} \nonumber \\ 
&\leq& \left(1 + \frac{t^{2}}{d^{m - l - 1}}  \right)
       \frac{t^{2}}{d^{2l}} + \frac{5t^{2}}{d^{l}} \leq
       \frac{6t^{2}}{d^{l}}. 
\end{eqnarray}
where the before-last inequality follows from \eq{columnsum-second} of \lemref{propertiesHnt} (which can be applied as $t^2 < d^l$). Then, choosing $l \geq \lceil  2.5 \log_d(4t) \rceil$ we find $M \leq (2\sqrt{l})^{-1}$, and we get \eq{shorten-second} from Lemma \ref{prop:gap-N}.

Let us turn to prove Eq. (\ref{blockdiagonalform}). Consider the linear map
\begin{equation}
B_k := \sum_{\pi \in \cS_t} \ket{\psi_{\pi, d}}^{\otimes k}\bra{\pi},
\end{equation}
with $\{ \ket{\pi} \}_{\pi \in \cS_t}$ an orthornomal set of vectors. Using Eq. \eq{normboundperm-second} of \lemref{propertiesHnt} we have
\begin{eqnarray} \label{boundBBdagger}
\Vert B_kB_k^{\cal y} \Vert_{\infty} = \left \Vert \sum_{\pi \in \cS_t}  \left( \ket{\psi_{\pi, d}}\bra{\psi_{\pi, d}} \right)^{\otimes k}  \right \Vert_{\infty} \leq 1 + \frac{t^{2}}{d^{k}}
\end{eqnarray}
Then
\begin{eqnarray}
 \left \Vert \sum_{\pi \in \cS_t} (\ket{\psi_{\pi, d}}\bra{\psi_{\pi, d}})^{\otimes (m - l - 1)} \otimes Y_{\pi} \right \Vert_{\infty} 
&=& \left \Vert \sum_{\pi \in \cS_t} (B_{(m - l - 1)} \ket{\pi}\bra{\pi}B_{(m - l - 1)}^{\cal y}) \otimes Y_{\pi} \right \Vert_{\infty} \nonumber \\
&\leq& \left \Vert  B_{(m - l - 1)}B_{(m - l - 1)}^{\cal y}  \right \Vert_{\infty} \left \Vert \sum_{\pi \in \cS_t} \ket{\pi}\bra{\pi}\otimes Y_{\pi} \right \Vert_{\infty} \nonumber \\
&=& \left \Vert  B_{(m - l - 1)}B_{(m - l - 1)}^{\cal y}  \right\Vert_{\infty}  \max_{\pi} \Vert Y_{\pi} \Vert_{\infty},
\end{eqnarray}
and Eq. (\ref{blockdiagonalform}) follows from the bound given by Eq. (\ref{boundBBdagger}).
\end{proof}

\subsection{Proof of Lemma \ref{lconvergenceWasserstein}} \label{proofLemma3} 
%\label{sec:coupling}
For two probability distributions $\nu_1, \nu_2$, we say $(X, Y)$ is a coupling for
$\nu_1, \nu_2$ if $X$ and $Y$ are distributed according to $\nu_1$ and $\nu_2$,
respectively.  
Define the $L^{p}$ Wasserstein distance between two probability distributions $\nu_1$ and
$\nu_2$ with respect to a distance measure $d$ as follows
\begin{equation}
  W_{d,p}(\nu_1, \nu_2) :=
  \inf \left \{  \E[d(X, Y)^{p}]^{1/p} \hspace{0.1 cm} : \hspace{0.1 cm} (X, Y)
    \hspace{0.1 cm}
    \text{is a pair of random variables coupling} \hspace{0.1 cm} (\nu_1, \nu_2)   \right \}.
\end{equation}
We note that
\begin{equation} \label{relationWs}
W_{\text{Fro}}(\nu_1, \nu_2) \leq  W_{\text{Rie},1}(\nu_1, \nu_2) \leq W_{\text{Rie},2}(\nu_1, \nu_2).
\end{equation}

We now state Oliveira's result (in fact a particular case of Theorem 3 of \cite{Oli07}), which offers a version of the path coupling method for Markov chains on the unitary group. It shows that a local contraction, in the $L^{2}$ Wasserstein distance, can be boosted into a global contraction.

\begin{lem}[Oliveira, Theorem 3 of \cite{Oli07}] \label{olilemma}
Let $\nu$ be a probability measure on $\mathbb{U}(d)$ such that
\begin{equation}
\limsup_{\varepsilon \rightarrow 0} \sup_{X,Y \in \mathbb{U}(d)} \left\{
  \frac{W_{\text{Rie},2}(\nu \ast \delta_{X}, \nu \ast \delta_{Y})}{d_{\text{Rie}}(X,Y )}
  \hspace{0.1 cm} : \hspace{0.1 cm}
  d_{\text{Rie}}(X,Y) \leq \varepsilon \right \} \leq \eta,
\end{equation}
with $\delta_{U}$ a mass-point distribution at $U \in \mathbb{U}(d)$. Then for all probability measures $\nu_1, \nu_2$ on $\mathbb{U}(d)$,
\begin{equation} 
W_{\text{Rie},2}(\nu \ast \nu_1, \nu \ast \nu_2) \leq \eta W_{\text{Rie},2}(\nu_1, \nu_2) .
\end{equation}
\end{lem}

%We shall use path coupling method from \cite{Oli07}.
%He proved convergence of Kac walk to Haar measure.
%We consider convergence of random circuits to Haar measure, and we copy his approach, 
%which eventually boils down to compute shrinking factor 
%of two infinitesimally close points subjected to a single step 
%with a suitable coupling. It turns out though, that in our case, 
%we get a nontrivial contraction only while considering $n-1$ steps of walk 
%for  $n$ systems (i.e. if Hilbert space is given by $\hcal=(\C^d)^{\ot n}$. 
In the rest of this section we apply Lemma \ref{olilemma} to prove Lemma \ref{lconvergenceWasserstein}.
Before we turn to the proof of Lemma \ref{lconvergenceWasserstein} in earnest, we prove the particular case of the random walk on three sites. Then in the sequence we will built up on it to get the general case.

\begin{lem}  \label{lem:couplingfor3}
For every integer $k > 0$,
\begin{equation}
  W_{\text{Rie},2}(\nu_{\text{LR}, 3, d})^{* 2k}, \mu_{\text{Haar}})
  \leq \left(1 - \frac{1}{2d^2+2} \right)^{k/2}\sqrt{2}d^{3/2}.
\end{equation}
\end{lem}

\begin{proof}
We will show
\be \label{maincondtislemma}
\limsup_{\varepsilon \rightarrow 0} \sup_{X,Y \in \mathbb{U}(d^{3})} \left\{
  \frac{W_{\text{Rie},2}((\nu_{\text{LR}, 3, d})^{* 2} \ast \delta_{X}, (\nu_{\text{LR},
      3, d})^{* 2} \ast \delta_{Y})} {d_{\text{Rie}}(X,Y )} \hspace{0.1 cm} : \hspace{0.1 cm}
  d_{\text{Rie}}(X,Y) \leq \varepsilon \right \}  \leq \eta := \sqrt{1 - \frac{1}{2d^{2} + 2}}.
\ee
Then applying Lemma \ref{olilemma} repeatedly we find
\begin{eqnarray} \label{aux8655}
  W_{\text{Rie},2}((\nu_{\text{LR}, 3, d})^{* 2k}, \mu_{\text{Haar}})
  &=& W_{\text{Rie},2}((\nu_{\text{LR}, 3, d})^{* 2k}\ast \delta_{\I}, (\nu_{\text{LR}, 3,
      d})^{* 2k} \ast \mu_{\text{Haar}}) \nonumber \\
  &\leq& \eta^{k}W_{\text{Rie},2}(\delta_{\I}, \mu_{\text{Haar}}) \nonumber \\
  &\leq& \eta^{k} \pi \d^{3/2},
\end{eqnarray}
where in the last inequality we used that
$W_{\text{Rie},2}(\delta_{I}, \mu_{\text{Haar}})
\leq \max_{X,Y} d_{\text{Rie}}(X,Y ) \leq \pi d^{3/2}$.  For this last
inequality we use the fact that the optimal path is of the form $X e^{iHt}$ for
$\|H\|_\infty\leq \pi$~\cite{Andruchow} (see also Section 4.4 of \cite{Oli07} for a
similar argument.)
The statement of the lemma thus follows from Eqs. (\ref{aux8655}) and  (\ref{relationWs}).

Let us turn to prove Eq. (\ref{maincondtislemma}). Let $X$ and $Y$ be two unitaries acting
on three $d$-dimensional systems satisfying $d_{\text{Rie}}(X,Y)\leq \eps$.
Consider two steps of the walk. Then we have four possibilities, each occuring with probability $\frac14$,
\ben
&&X \to  \left \{ \tilde U_{12} U_{12}X, \tilde U_{23} U_{12}X, \tilde U_{12} U_{23}X, \tilde U_{23} U_{23}X \right \},
\label{eq:trivial1}
\een
for independent Haar distributed unitaries $U_{12}, U_{23}, \tilde U_{12}, \tilde U_{23}$, and likewise for $Y$. Here the indices of the unitaries label in which subsystems they act non-trivially. 

At the moment we have a trivial coupling, i.e. $X$ and $Y$ are subjected to the same transformation. Now we introduce a nontrivial coupling, which we show on average brings two infinitesimally close unitaries closer to each other. We consider the transformation:
\ben
&&X \to  X' \in \left \{ \tilde U_{12} U_{12}X, \tilde U_{23} V_{23} U_{12}X, \tilde U_{12} V_{12} U_{23}X, \tilde U_{23} U_{23}X \right \} 
\label{eq:trivial2}
\een
where the unitary $V_{23}$ can depend on $U_{12}$ and $V_{12}$ can depend on $U_{23}$, and of course both can depend on $X$ and $Y$. The unitary $Y$, in turn, undergoes the same transformation as before, namely
\ben
&&Y \to  Y' \in \left \{ \tilde U_{12} U_{12}Y, \tilde U_{23} U_{12}Y, \tilde U_{12} U_{23}Y, \tilde U_{23} U_{23}Y \right \} 
\label{eq:trivial3}
\een

Let us check that the transformations above indeed define a valid coupling. In order to do so the induced distribution on the two unitaries $X'$ and $Y'$  must be the same as in the case of a trivial coupling. This is clearly true for $Y'$. To see that it is also true for $X'$, we observe that for any fixed $V_{23}$, $\tilde U_{23} V_{23}$ is Haar distributed for a Haar distributed $\tilde U_{23}$ (and likewise for $\tilde U_{12} V_{12}$). 
 
In the sequel we show
\be
\E \left( d_{\text{Rie}}(X',Y')^2 \right) \leq \eta^{2} \E \left( d_{\text{Rie}}(X,Y)^2
\right) + O(\eps^3),
\label{eq:shrink}
\ee
where $X'$ and $Y'$ are random variables related by the coupling and $X$ and $Y$ satisfy
$d_{\text{Rie}}(X,Y) \leq \eps$. We first estimate both sides with respect to the Frobenius
distance. The LHS of \prettyref{eq:shrink} then corresponds to
\ben
&&\E \left(\|X'-Y'\|_2^2 \right) =
\frac14 \left(\E \left( \|\tilde U_{12} U_{12} X - \tilde U_{12} U_{12} Y\|_2^2 \right) + 
\E \left( \|\tilde U_{23} V_{23} U_{12} X- \tilde U_{23} U_{12} Y\|_2^2 \right) + \right. \nonumber \\
&& \left. \E \left(\|\tilde U_{12} V_{12} U_{23} X- \tilde U_{12} U_{23} Y\|_2^2 \right) +  
\E \left( \|\tilde U_{23}  U_{23} X- \tilde U_{23} U_{23} Y\|_2^2 \right) \right),
\een
with the expectation taken over Haar distributed $\tilde U_{12},
U_{12}, \tilde U_{23}, U_{23}$. Using the unitary invariance of the
2-norm we can rewrite this as  
\be
\E \left(\|X'-Y'\|_2^2 \right) = 
\frac14 \left(2\|X-Y\|_2^2 + \E \left( \|V_{23} U_{12} X- U_{12} Y\|_2^2 \right) + \E \left( \|V_{12} U_{23} X- U_{23} Y\|_2^2 \right) \right).
\label{eq:norm-walk1}
\ee
Since $V_{12}$ and $V_{23}$ can depend in an arbitrary way on $U_{23}$ and $U_{12}$, respectively, we can take the minimum over $V_{12}$ and $V_{23}$ to get
\begin{eqnarray} \label{eq58}
\E \left(\|X'-Y'\|_2^2 \right)= 
\frac14 \left( 2\|X-Y\|_2^2 \right. &+&  \E \left(  \min_{V_{23}}\|V_{23} U_{12} X- U_{12} Y\|_2^2 \right)   \nonumber \\ &+& \left. \E \left(  \min_{V_{12}} \|V_{12} U_{23} X- U_{23} Y\|_2^2 \right) \right).
\label{eq:norm-walk2}
\end{eqnarray}

For any two unitaries $U_1, U_2$ we have 
\be
\Vert U_1-U_2 \Vert^2_2=2\left(\tr(\I) - \text{Re}\left(\tr(U_1U_2^\dagger) \right) \right).
\label{eq:norm-scalar}
\ee
Since $X$ and $Y$ are infinitesimally close we can write
\be
R := XY^\dagger = e^{i\ep H}=\I + i \ep H -\frac{\ep^2}{2} H^2  + O(\ep^3)
\label{eq:W-ep}
\ee
for a Hermitian matrix $H$ with $\Vert H \Vert_{2} \leq 1$. Then applying \eqref{eq:norm-scalar} we get 
\be
\|X-Y\|_2^2 = \ep^2 \tr(H^2) + O(\ep^3)
\ee

Let us now consider the term 
$\E \left( \min_{V_{12}}\|V_{12} U_{23} X- U_{23} Y\|_2^2 \right)$ (for the other term the calculation gives the same result).
We have 
\begin{eqnarray} \label{tracenormappears}
\E \left(\min_{V_{12}}\|V_{12} U_{23} X- U_{23} Y\|_2^2\right) &=& 
2\left(\tr(\I) -\E \left( \max_{V_{12}} |\tr (V_{12} U_{23} X Y^\dagger U_{23}^\dagger)| \right) \right) \nonumber \\
&=& 2\left( \tr(\I) - \E \|\tr_3(U_{23} R  U_{23}^\dagger)\|_1 \right),
\label{eq:average-single path-3}
\end{eqnarray}
with $R=XY^\dagger$. The last equality follows from the following variational
characterizarion of the trace norm: $\Vert Z \Vert_1 = \max_{U \in \mathbb{U}} |\tr(UZ)|$ \cite{bhatia2013matrix}.

From Eq. \eqref{eq:W-ep}  we get 
\be
\tr_3(U_{23} R  U_{23}^\dagger)= d\, \I_{12} + i \ep \tr_3 (U_{23} H U_{23}^\dagger)
- \frac{\ep^2}{2} \tr_3 (U_{23} H^2 U_{23}^\dagger) + O(\ep^3).
\ee
An easy calculation shows that for any two Hermitian operators $A,B$ we have 
\be
\left\|\I + i\ep A - \frac{\ep^2}{2} B\right\|_1 = \tr(\I) +\frac{\ep^2}{2}(\tr A^2 - \tr B) + O(\ep^3).
\ee
Hence we obtain 
\be
\left\|\tr_3(U_{23} R  U_{23}^\dagger)\right\|_1= 
\tr(\I) +\frac{\ep^2}{2}\frac1d \tr\bigl((\tr_3 (U_{23} H U_{23}^\dagger))^2\bigr) 
- \frac{\ep^2}{2}\tr(H^2) + O(\ep^3),
\ee
so that by Eq. (\ref{tracenormappears}),
\be
\E \inf_{V_{12}} \|V_{12} U_{23} X- U_{23} Y\|_2^2 = 
\ep^2 \left[ \tr(H^2) - \frac1d \E \left( \tr \bigl( (\tr_3 (U_{23} H U_{23}^\dagger))^2 \bigr)\right) \right]+ O(\ep^3).
\label{eq:inf-E}
\ee

Now our goal is to compute the average $\E \left( \tr \bigl((\tr_3 (U_{23} H U_{23}^\dagger))^2\bigr) \right)$. 
We note that for any operator $C_{123}$ we have 
\be
\tr( C_{12}^2) = \tr \L(( C_{123}\ot C_{\overline{123}})\,
(\F_{12:\overline{12}}\ot \I_{3\overline{3}}) \R)
\ee
where systems with bars are copies of original systems, and $\F$ is the operator which 
swaps systems $12$ with $\overline{12}$.  Therefore
\ben
\E \left( \tr \bigl(\tr_3 (U_{23} H U_{23}^\dagger)^2\bigr) \right)
=\E \left( \tr \bigl( (H_{123}\ot H_{\overline{123}})\,(U_{23}^\dagger\ot U_{\overline{23}}^\dagger)
(\F_{12:\overline{12}}\ot \I_{3\overline{3}}) (U_{23}\ot U_{\overline{23}}) \bigr) \right). 
\een
We now compute (see also Lemma IV.3 of \cite{ADHW06})
\be
\E \left( (U_{23}^\dagger \ot U_{\overline{23}}^\dagger)( \F_{2:\overline{2}}\ot \I_{3\overline{3}})( 
U_{23}\ot U_{\overline{23}})\right)= 
\frac{d}{d^2+1} (I_{23:\overline{23}} + \F_{23:\overline{23}})
\label{eq:Flip-av}
\ee
Using the fact that the tensor product of swap operators is again a swap operator (e.g. $\F_{12:\overline{12}}=\F_{1:\overline{1}}\ot \F_{2:\overline{2}}$), we obtain 
\begin{eqnarray}
  \E \left( \tr \bigl(\tr_3 (U_{23} H U_{23}^\dagger)^2\bigr) \right) 
  &=& 
      \frac{d}{d^2+1} \left( \tr \bigl( (H_{123}\ot H_{\overline{123}})\,
      \F_{123:\overline{123}}\bigr) 
      +  \tr \bigl( (H_{123}\ot H_{\overline{123}})\, \F_{1:\overline{1}}
      \otimes \I_{23, \overline{23}} \bigr) \right)
      \nonumber \\ &=& \frac{d}{d^2+1} \left( \tr (H^2) + \tr(H_1^{2}) \right) \nonumber \\ 
  &\geq&  \frac{d}{d^2+1} \tr (H^2).
\label{eq:Flip-av2}
\end{eqnarray}
Inserting this into \eqref{eq:inf-E} 
\be
\E  \inf_{V_{12}}\left( \|V_{12} U_{23} X- U_{23} Y\|_2^2 \right) \leq \ep^2 \left(1 - \frac{1}{d^2+1} \right)\tr (H^2)+O(\ep^3).
\ee
Finally, using Eq. (\ref{eq58}),
\be
\E \left(\|X'-Y'\|_2^2 \right) \leq \ep^2 \left(\frac{2d^2 +1}{2d^2+2}\right)
\tr(H^2)+O(\ep^3)
= \eta^2 \E\left( \|X-Y\|_2^2\right) +O(\ep^3).
\label{eq:shrink-Fro}\ee
This proves the desired bound in \prettyref{eq:shrink} except with respect to the
Frobenius distance and not the Riemannian distance.  However, these distances are the same to
same leading order for very close points.  Indeed \cite{Oli07} shows
\be d_{\text{Fro}}(U,V) \leq d_{\text{Rie}}(U,V) \leq
d_{\text{Fro}}(U,V)  + c d_{\text{Fro}}(U,V)^2 .\ee
For our purposes we would like $c$ to be independent of dimension.  To see this, consider
for simplicity $U=I$ and $V=e^{i\diag(\lambda)}$ for $\lambda =
(\lambda_1,\ldots,\lambda_D)$, $D=d^3$, and with $\diag(\lambda)$ the matrix with
$\lambda$ along the diagonal and zeros elsewhere.  Then we compute (using \cite{Andruchow})
\ba
d_{\text{Fro}}(U,V)^2 &= \sum_j |1 - e^{i\lambda_j}|^2 = 4\sum_j \sin^2(\lambda_j/2) \\
d_{\text{Rie}}(U,V)^2 &= \sum_j \lambda_j^2
\ea
Using $\sin x  \geq x - x^3/6$ we have $\sin^2(x) \geq x^2 - x^4/3$ and
\be d_{\text{Fro}}(U,V)^2  = 4\sum_j \sin^2(\lambda_j/2) \geq
\sum_j \lambda_j^2 - \frac{1}{12} \sum_j \lambda_j^4
\geq d_{\text{Rie}}(U,V)^2 - \frac{1}{12} d_{\text{Rie}}(U,V)^4
\ee
This allows to convert \prettyref{eq:shrink-Fro} into our desired bound:
\begin{subequations}\ba \E\left(d_{\text{Rie}}(X',Y')^2\right)
&\leq \E\left(d_{\text{Fro}}(X',Y')^2\right) + 12\eps^4 \\
&\leq \eta^2 \E\left(d_{\text{Fro}}(X,Y)^2\right) + O(\eps^3)\\
&\leq \eta^2 \E\left(d_{\text{Rie}}(X,Y)^2\right) + O(\eps^3)
\ea \end{subequations}
\end{proof}

\vspace{0.2 cm}

\noindent \textbf{Remark} (\textit{Why one step of the walk does not work}): It is instructive to see why coupling only one step of the walk does not seem to be enough to prove contraction. In this case, a general class of couplings is given by
\ben
&& X \to \{ U_{12} V_{12} X,  U_{23} V_{23} X \} \nonumber \\
&& Y \to \{  U_{12} Y ,U_{23} Y \},
\een
where $V_{12}$ and $V_{23}$ can depend only on $X$ and $Y$.
If we optimize over the choice of $V_{12}, V_{23}$ we get 
\begin{subequations}\begin{align}
\E \left(\|X-Y\|^2_2 \right) &=  2\tr(\I) - \E \left( \|\tr_3(R)\|_1 \right)  - \E \left(
  \|\tr_1(R)\|_1 \right) \\
&= \ep^2\left(\tr(H^2) - \frac1{2d}(\tr(H_{12}^2)+\tr(H_{23}^2)  \right) + O(\ep^{3})
\end{align}\end{subequations}
where $H_{12} = \tr_3 (H)$, and $H_{23} = \tr_1(H)$. However there exist Hermitian matrices $H$ such that 
$H_{12}=H_{23}=0$, in which case $\E \|X-Y\|^2_2=\ep^2 \tr(H^2) + O(\ep^{3}) = \|X_0-Y_0\|^2_2$, and
so that we do not have any contraction. We can thus understand the role of the second step of the walk in constructing a useful coupling: it is to randomly change such bad cases of $H$ into good $H$, with non-zero probability. 
$\square$

We are now in position to prove Lemma \ref{lconvergenceWasserstein}:

\noindent \textbf{Lemma \ref{lconvergenceWasserstein} (restatement)}
\textit{For every integers $k, n > 0$,}
\begin{equation}
  W_{\text{Rie},2}((\nu_{\text{LR}, n, d})^{* (n - 1)k}, \mu_{\text{Haar}})
  \leq  \left( 1 - \frac{1}{e^{n}(d^2+1)^{n-2}} \right)^{\frac{k}{n-1}} \pi d^{n/2}.
\end{equation}

\begin{proof}

We will show that 
\be  \label{nstepsW}
\limsup_{\varepsilon \rightarrow 0} \sup_{X, Y \in \mathbb{U}(d^{n})} \left\{ \frac{W_{\text{Rie},2}((\nu_{\text{LR}, n, d})^{* (n - 1)} \ast \delta_{X}, (\nu_{\text{LR}, n, d})^{* (n - 1)} \ast \delta_{Y})}{\Vert X - Y \Vert_{2}} \hspace{0.1 cm} : \hspace{0.1 cm} \Vert X - Y \Vert_{2} \leq \varepsilon \right \}  \leq \eta,
\ee
with
\be
\eta := \left( 1 - \frac{1}{e^{n}(d^2+1)^{n-2}} \right)^{\frac{1}{n-1}}.
\ee
Then, in analogy to the proof of Lemma \ref{lem:couplingfor3},
\begin{eqnarray} 
W_{\text{Rie},2}((\nu_{\text{LR}, n, d})^{* (n-1)k}, \mu_{\text{Haar}}) &=& W_{\text{Rie},2}((\nu_{\text{LR}, n, d})^{* (n-1)k}\ast \delta_{\I}, (\nu_{\text{LR}, n, d})^{* (n-1)k} \ast \mu_{\text{Haar}}) \nonumber \\ &\leq& \eta^{k}W_{\text{Rie},2}(\delta_{\I}, \mu_{\text{Haar}}) \nonumber \\ &\leq& \eta^{k} \pi d^{n/2},
\end{eqnarray}
and the statement of the lemma follows from the bound
$W_{\text{Rie},2}((\nu_{\text{LR}, n, d})^{* (n-1)k}, \mu_{\text{Haar}})
\leq W_{\text{Rie},2}((\nu_{\text{LR}, n, d})^{* (n-1)k}, \mu_{\text{Haar}})$.

Let us turn to prove Eq. (\ref{nstepsW}). In order to avoid the problem that occured when we applied a single step of 
the walk to three systems, we now need to apply $k=n-1$ steps of walk. There are then $k^k$ possible paths, and we make a nontrivial coupling 
only for $k!$ of them. Namely, for those paths for which no pair of systems is repeated, i.e. for the case $U_{n-1 n}\ldots U_{23} U_{12}$ 
and all its permutations (all sequences which come from permuting the order of the unitaries in the sequence above). For those $k!$ paths we consider the following coupling 
\ben
&&X \to X' := U_{i_{n-1},i_{n-1}+1} V_{i_{n-1},i_{n-1}+1} \ldots U_{i_2,i_2+1} U_{i_1,i_1+1}X  \nonumber \\
&&Y \to Y' := U_{i_{n-1},i_{n-1}+1} V_{i_{n-1},i_{n-1}+1} \ldots U_{i_2,i_2+1} U_{i_1,i_1+1} Y
\een
where $V$ can depend on all unitaries sitting to the right, and $i_j \in \{1,\ldots,n-1 \}$. 
We now consider explicitly a particular sequence $U_{12} U_{23} \ldots U_{n-1 n}$
and compute the analogue of \eqref{eq:average-single path-3} (for the other sequences 
the calculations give the same result). We have
\begin{eqnarray}
&&\inf_{V_{12}}\E \left( \|U_{12} V_{12} U_{23}\ldots U_{n-1 n} X- U_{12} U_{23}\ldots U_{n-1 n} Y\|_2^2 \right) \nonumber \\
&=& 2\left( \tr \left(\I \right) - \E \left( \|\tr_{3\ldots n} (U_{23}\ldots U_{n-1 n} R  U_{23}^\dagger\ldots U_{n-1 n}^\dagger)\|_1 \right) \right).
\end{eqnarray}
Expanding in $\ep$ we get in analogy to \eqref{eq:inf-E}:
\ben
&&\inf_{V_{12}}\E\|U_{12} V_{12} U_{23}\ldots U_{n-1 n} X- U_{12} U_{23}\ldots U_{n-1 n} Y\|_2^2=
\nonumber \\
&&\ep^2 \left[ \tr(H^2) - \frac{1}{d^{n-2}} \E \tr \bigl( (\tr_{3\ldots n} (U_{23}\ldots U_{n-1 n} H U_{23}\ldots U_{n-1 n}^\dagger))^2\bigr)\right]+ O(\ep^3).
\label{eq:inf-E-n}
\een
Moreover, using repeatedly \eqref{eq:Flip-av2}
we obtain 
\be
\E \left(\tr \bigl( (\tr_{3\ldots n} (U_{23}\ldots U_{n-1 n} H U_{23}\ldots U_{n-1 n}^\dagger))^2\bigr) \right)
\geq \left(\frac{d}{d^2+1}\right)^{n-2}\tr \left(H^2 \right),
\ee
so that 
\be
\inf_{V_{12}}\E \left( \|U_{12} V_{12} U_{23}\ldots U_{n-1 n} X- U_{12} U_{23}\ldots U_{n-1 n} Y\|_2^2 \right)
\leq \ep^2 \left(1- \frac{1}{(d^2+1)^{n-2}} \right)\tr(H^2)  + O(\ep^3).
\ee
Finally, we have $k^k-k!$ paths of walk for which we do not have any shrinking (as our coupling was trivial 
for those paths) and $k!$ paths in which we have a shrinking factor of $1 - (d^2+1)^{-(n-2)}$. 
Thus this gives 
\be
\E \left(\|X'-Y'\|^2_2 \right) \leq \ep^2 \chi \|X-Y\|_2^2 +O(\ep^3)
\ee
with
\be
\chi := 1-\frac{(n-1)!}{(n-1)^{(n-1)}} \frac{1}{(d^2+1)^{n-2}} \leq 1 - \frac{1}{e^{n}(d^2+1)^{n-2}} = \eta,
\ee
where we used the bound $n! \geq n^{n} e^{-n}$.

We now can convert this bound to one for the Riemannian distance using the same arguments
used in \prettyref{lem:couplingfor3}.
\end{proof}

\subsection{Proof of Lemma \ref{lem:WasserGap}} \label{proofLemma4}

In this section we prove the last lemma needed in the proof of \thmref{main-TPE}:

\noindent
\begin{replem}{lem:WasserGap}
\lemWasserGap
\end{replem}

\begin{proof}

%For a fixed $t$, the convergence rate is given by  
%$\Delta(N)=1-\lambda_{t!+1}(N)$  where $\lambda_{t!+1}$ is 
%the second largest eigenvalue  of the superoperator 
%$\frac 1N \int \mu({\rm d}U) U^{t, t}$, with $U^{t, t}=
%U^{\otimes t} \otimes \overline{U}^{\otimes t}$ 
%acting on $N$ systems.  (For the largest one we have $\lambda_1=1$).
%We shall need the following proposition, connecting the value of 
%the above convergence rate with the shrinking factor $\eta$  obtained in section 
%\ref{sec:coupling}, in \eqref{eq:eta}. 
%The gap for superoperator \eqref{someequation} (Hamiltonian divided by n) and 
%the convergence rate $\eta$ are equal 

%From Lemma \ref{lem:Ham-gap} we have
The definition of $g(\cdot,\cdot)$ states that
\begin{eqnarray}
g(\nu, t) = \left\Vert \int_{\mathbb{U}(d^{n})} U^{\otimes t, t} \nu({\rm d}U) - \int_{\mathbb{U}(d^{n})} U^{\otimes t, t} \mu_{\text{Haar}}({\rm d}U)  \right\Vert_{\infty}.
\end{eqnarray}
Let $X$ be such that $\|X\|_1 \leq 1$ and
\begin{eqnarray}
   && \tr \left( \left(\int \nu({\rm d}U) U^{\otimes t, t} - \int \mu_{\text{Haar}}({\rm d}U) U^{\otimes t, t} \right)X\right) \nonumber \\
&=& \left\Vert \int_{\mathbb{U}(d^{n})} U^{\otimes t, t} \nu({\rm d}U) - \int_{\mathbb{U}(d^{n})} U^{\otimes t, t} \mu_{\text{Haar}}({\rm d}U)  \right\Vert_{\infty}.
\label{eq:norm_vs_X}
\end{eqnarray}
That such a $X$ always exists follows from the following variational characterization of the operator norm: $\Vert A \Vert_{\infty} = \max_{X} \{ \tr(AX) \hspace{0.1 cm} : \hspace{0.1 cm} \Vert X \Vert_1 \leq 1$ \} \cite{bhatia2013matrix}.

Define $f(U) := \tr(U^{\otimes t, t}X)$. We claim $f$ is $2t$-Lipschitz. Before proving it, let us show how it implies the statement of the lemma. Indeed, since $f/(2t)$ is 1-Lipschitz,
\ben
g(\nu, t) &=&  2t \left| \int f(U)/(2t) \nu({\rm d}U) - \int f(U)/(2t) \mu_{\text{Haar}}({\rm d}U) \right| \nonumber \\
  &\leq&  2t W_{\text{Fro}}(\nu, \mu_{\text{Haar}}),
  \label{eq:H_vs_W}
\een
where the last inequality follows from the definition of the Wasserstein distance, given by Eq. (\ref{WassersteinDef}).

It remains to show that $f(U)$ is $2t$-Lipschitz. This follows from
\begin{eqnarray}
|f(U) - f(V)| &=& |\tr((U^{\otimes t, t} - V^{ \otimes t, t})X)| \nonumber \\
              &\leq& \|X\|_1 \|U^{ \otimes t, t} - V^{ \otimes t, t}\|_{\infty} \nonumber \\
             &\leq& \|U^{\otimes t, t} - V^{\otimes t, t}\|_{\infty} \nonumber \\
             &\leq&  2t \|U - V\|_{\infty} \nonumber \\
             &\leq& 2t\|U - V\|_{2}
\end{eqnarray}
The first inequality follows from the relation $\tr(A^{\cal y}B) \leq
\Vert A \Vert_1 \Vert B \Vert_{\infty}$ and the second from the bound
$\Vert X \Vert_1 \leq 1$, and the third from the hybrid argument; that
is,
by repeatedly applying the inequality
\be
\Vert A \otimes B - C \otimes D \Vert_\infty \leq \Vert A - C \Vert_\infty + \Vert B - D\Vert_\infty,
\ee
valid for unitaries $A, B, C$ and $D$. This, in turn, follows from
\begin{eqnarray}
\label{eq:hybrid}
\Vert A \otimes B - C \otimes D \Vert_\infty &=& \Vert A \otimes (B - D) + (A - C) \otimes D \Vert_\infty \nonumber \\ &\leq& \Vert A \otimes (B - D)\Vert_{\infty} + \Vert (A - C) \otimes D \Vert_\infty  \nonumber \\ &\leq&  \Vert A \Vert _{\infty}\Vert B - D \Vert_\infty +  \Vert D \Vert _{\infty}\Vert A - C\Vert_\infty \nonumber \\  &\leq&  \Vert B - D \Vert_\infty + \Vert A - C\Vert_\infty.
\end{eqnarray}
\end{proof}

%Finally, we obtainthe length of local random circuit needed 
%to obtain $\epsilon$ 
%\begin{thm}

%The gap of our Hamiltonian amounts to 
%\be
%Gap=\eta_1(n)
%\ee
%\end{thm}

%{\bf Proof.}
%Combining \eqref{eq:eta_lambda} with \eqref{eq:gap_Nl_2} 
%and the relation between Gap of H and $\lambda_{t!+1}$
%\be
%Gap(n)=n(1-\lambda_{t!+1})
%\ee
%which follows from the fact, that the gap $G(n)$ is computed on $nM$ 
%while $\lambda_{t!+1} $ is eigenvalues of $M$. 

\section{Proof of Corollary \ref{cor:unisetofgates}}
\label{sec:proof-univ}

We now show how we can get convergence rates for other universal set
of gates from our analysis of the Haar random case:
\restatecor{cor:unisetofgates}
%\begin{replem}{cor:unisetofgates}\corunisetofgates\end{replem}

\begin{proof}
Define the Hermitian matrices
\begin{equation}
  P_{G, t} := \frac{1}{m}\sum_{i=1}^m g_{i}^{\otimes t, t}
  \qquad
  P_{\text{Haar},t} := 
  \int_{\mathbb{U}(d^2)} U^{\otimes t, t} \mu_{\text{Haar}}({\rm d}U)
\end{equation}
Note that  $P_{\text{Haar},t}$ is a projector while $P_{G,t}$ in general will not be.  The
two operators commute and indeed $P_{G,t}  P_{\text{Haar},t} =   P_{\text{Haar},t}$.  Thus
$P_{G,t}$ acts trivially on the subspace projected on by $  P_{\text{Haar},t}$.  The only
question is then what it does with the orthogonal space.

It was proven in \cite{BG11} that there is a constant $\lambda < 1$, independent of $t$, such that for all $t$,
\be \label{eq:BG}
\left \Vert  P_{G, t} -
  P_{\text{Haar},t} \right \Vert_{\infty} \leq \lambda.
\ee
(Their breakthrough was to make $\lambda$ independent of $t$.  The same result has long
been known~\cite{AK62} with $\lambda$ a function of $G$ and $t$ for any
universal set of gates.) Using the block decomposition of $P_{G,t}$ we then have
\begin{equation}  \label{eq:aux876}
 I - P_{G, t}  \geq \lambda  \left(I -   P_{\text{Haar},t} \right).
\end{equation}

Define the local Hamiltonian in $\left(\mathbb{C}^{d} \right)^{\otimes n}$:
\be
H_{G, n, t} := \sum_{i=1}^{n-1} \left( I - P_{G, t}  \right)_{i, i+1}.
\ee 
Note it has the same groundstate as $H_{n, t}$. From Eq. (\ref{eq:aux876}):
\be
H_{G, n, t} \geq \lambda H_{n, t},
\ee
with $H_{n, t}$ given by Eq. (\ref{Hntdef}). Moreover $H_{G, n, t}$ has the same ground space as $H_{n, t}$. Thus
\be
\Delta(H_{G, n, t}) \geq \lambda  \Delta(H_{n, t}).
\ee
In complete analogy to the proof of \prettyref{lem:Ham-gap}, we have that
\begin{equation}
g(\nu_{\text{LR},n,d,G}, t) = 1 - \frac{\Delta(H_{G, n, t})}{n}
\end{equation}
with $\nu_{\text{LR},n,d,G}$ the distribution over circuits that is obtained by replacing
the Haar measure in $\nu_{\text{LR},n,d}$ with uniform measure over the set $G$.   (We
define $\nu_{\text{PLR},n,d,G}$  analogously.) The
first claim of the corollary now follows from Eq. (\ref{eq:boundt4logt}).

For the model with parallel gates, we need to replace $M_{n,t}$ in \prettyref{eq:Mnt-def}
with
\ba
M_{G,n,t} &:=
\frac{(P_{G,t})_{1,2}  (P_{G,t})_{3,4} \cdots (P_{G,t})_{n-1,n}
  + (P_{G,t})_{2,3} (P_{G,t})_{4,5}\cdots (P_{G,t})_{n-2,n-1}}{2}
\\ &:= \frac{P_{G,t,\text{odd}} + P_{G,t,\text{even}}}{2}
\ea
where we use the notation $(X)_{i,i+1}$ to indicate the operator $X$ acting on sites
$i,i+1$ tensored with identity operators elsewhere.
Let $P_{G,t}' = P_{G,t} - P_{\text{Haar},t}$ and observe that by \prettyref{eq:BG},
$\|P_{G,t}'\|_\infty \leq \lambda$.  Then
\ba
P_{G,t,\text{odd}}
 & = (P_{\text{Haar},t} + P_{G,t}')_{1,2}  (P_{\text{Haar},t} + P_{G,t}')_{3,4} \cdots
 (P_{\text{Haar},t} + P_{G,t}')_{n-1,n}
 \\
 & := P_{\text{odd}} + Q_{\text{odd}}
\ea
where we have defined
$P_{\text{odd}} =(P_{\text{Haar},t} )_{1,2}(P_{\text{Haar},t})_{3,4}\cdots
(P_{\text{Haar},t} )_{n-1,n}$ (as in \prettyref{sec:PLR}) and defined $Q_{\text{odd}}$ to
be the sum of the remaining $2^{n/2}-1$ terms.  Because of the block decomposition of
$P_{G,t}$ each of these terms is orthogonal to the others as well as to
$P_{\text{Haar},t}$.  Thus $\|Q_{\text{odd}}\|_\infty \leq \lambda$ and $Q_{\text{odd}}$
and $P_{\text{odd}}$ are orthogonal to each other.  This means that
\be P_{G,t,\text{odd}} \leq P_{\text{odd}}  + \lambda (I - P_{\text{odd}})
 = (1-\lambda) P_{\text{odd}} + \lambda I.\ee
 Similar arguments imply that
\be P_{G,t,\text{even}} \leq P_{\text{even}}  + \lambda (I - P_{\text{even}})
= (1-\lambda) P_{\text{even}} + \lambda I.\ee
Recall from \prettyref{lem:detectabilityconsequence} and the following discussion that
\be
\frac{P_{\text{odd}}  + P_{\text{even}}}{2} \leq P_{\text{all,t}} +
(1-\delta/13) (I - P_{\text{all},t}),
\ee
where $\delta$ is the lower bound on $\Delta(H_{n,t})$ from \prettyref{eq:boundt4logt}.
Putting this together we find that
\be
 \lambda_{t!+1}(M_{G,n,t}) \leq 1 - \frac{(1-\lambda)\delta}{13}.
%M_{G,n,t} \leq (1-\lambda) P_{all,t} + (1-\lambda)(1-\delta) (
\ee
This bound on the gap can yield the desired bound on the number of gates needed using the
same arguments that have appeared for the case of Haar-random gates.
\end{proof}

\section{Proof of Proposition \ref{prop:converse}}

In this section, we state lower bounds on the size of $t$-designs that
match our results in \corref{main-design} up to polynomial factors.

First we argue that if $\nu$ is an approximate $t$-design, it must
have large support.  More precise lower bounds are known for exact
$t$-designs~\cite{RS09} and for approximate 2-designs \cite{ABW09}, but
for our purposes, it will be enough to determine the rate of scaling.
\begin{lem}\label{lem:t-design-size-LB}
If $\nu$ is an $\eps$-approximate $t$-design on $\bbU(N)$ then
\be|\supp(\nu)| \geq (1-\eps)\binom{N+t-1}{t}^2 \ee
\end{lem}
\begin{proof}
Let $S=\vee^t\bbC^N$ be the symmetric subspace of $(\bbC^N)^{\ot t}$
(c.f. \defref{sym} in Appendix A)
Define $\ket\varphi$ to be the maximally entangled state on $S \ot
S$.  Since $S$ is an irrep of $\bbU(N)$ under the action $U\mapsto
U^{\ot t}$, it follows that $(\Delta_{\mu_{\text{Haar}, t}}\ot \id)(\varphi)$ is
the maximally mixed state on $S \ot S$.  This has rank
$\binom{N+t-1}{t}^2$.  Thus, to approximate this state to within trace
distance $\eps$ requires a state of rank at least $(1-\eps)
\binom{N+t-1}{t}^2$.  Finally, $\rank((\Delta_{\nu,t}\ot
\id)(\varphi)) \leq |\supp(\nu)|$.
\end{proof}

To relate the cardinality of a design with the number of gates in a
quantum circuit, we need to discretize the set of all quantum
circuits.   We say that a set $X$ is an $\eps$-covering of $Y$ if for
all $y\in Y$, there exists an $x\in X$ with $d(x,y)\leq \eps$ for some
distance measure $d$.
\begin{lem}\label{lem:eps-net-circuits}
There exists an $\eps$-covering in diamond norm of size 
$ \leq \binom{n}{2}^r \left(\frac{10r}{\eps}\right)^{rd^4}$
for the set of circuits on $n$ qudits comprised of $\leq r$ two-qudit gates.
\end{lem}

Before proving the lemma, we note the following useful bound
from part 6 of Lemma 12 of \cite{AKN98} that applies to any unitaries $U,V$:
\be \|\ad_U -\ad_V\|_\diamond \leq 2\|U-V\|_\infty.
\label{eq:diamond-unitary}\ee

\begin{proof}
To describe an $\eps$-covering for circuits, it suffices to specify the
location of each gate and to approximate each gate to accuracy
$\eps/r$.  The former has $\binom{n}{2}^r$ choices; the latter
requires $r$ copies of a $\eps/r$-covering for $U(d^2)$.  Finally,
standard arguments~\cite{MS86} show that such nets can be constructed with size
$\leq (5r/\eps)^{d^4}$ for the operator norm.  We convert operator
norm to diamond norm using \eq{diamond-unitary}.
\end{proof}

Combining these results we can prove that $t$-designs require large
circuits.
\restateprop{prop:converse}
%\begin{replem}{prop:converse}\converse\end{replem}

\begin{proof}
Let  the distribution $\nu$ be an $\varepsilon$-approximate unitary
$t$-design with all elements composed of $r$ two-qudit gates, possibly
including the identity.  From \lemref{eps-net-circuits}, construct a
diamond-norm $\delta$-covering for the set of circuits on
$n$ qudits, and denote it by $C_\delta$.
Consider a new distribution $\tilde \nu({\rm d}U) $
in which each unitary $U$ is replaced by its closest unitary $\tilde
U\in C_\delta$.
 We claim that $\{ \tilde \nu({\rm d}U), U \}$ is a $(\varepsilon
+ t \delta)$-approximate unitary $t$-design. Indeed  
\begin{subequations}\label{boundsanyway}\begin{align} 
\left \Vert \Delta_{\tilde \nu, t} - \Delta_{\mu_{\text{Haar}}, t}
\right \Vert_{\diamond}  
&\leq \left \Vert \Delta_{\nu, t} -
  \Delta_{\mu_{\text{Haar}}, t}    \right \Vert_{\diamond} + \left
  \Vert  \Delta_{\tilde \nu, t}  - \Delta_{\nu, t}   \right
\Vert_{\diamond} \\ 
 &\leq \varepsilon +  \max_{U\in\supp(\nu)}\min_{\tilde U\in
   C_\delta} \Vert \ad_{U^{\ot t}} - \ad_{\tilde U^{\ot t}}
 \Vert_{\diamond} \\   
 &\leq \varepsilon +  2\max_{U\in\supp(\nu)}\min_{\tilde U\in
   C_\delta} \| U^{\ot t} - \tilde U^{\ot t}\|_\infty
& \text{by \eq{diamond-unitary}}
\\
 &\leq \varepsilon +  2t\max_{U\in\supp(\nu)}\min_{\tilde U\in
   C_\delta} \| U - \tilde U\|_\infty 
& \text{Eq. \ref{eq:hybrid}}\\
&\leq  \varepsilon + 2t \delta, 
\end{align}\end{subequations}
See Fact 2.0.1 of \cite{BV93} for a statement and proof of the hybrid inequality.

Choosing $\delta =\varepsilon/2t$ we get that the distribution
$\tilde\nu$ is a $2\varepsilon$-approximate $t$-design.  Now we invoke
Lemmas \ref{lem:t-design-size-LB} and \ref{lem:eps-net-circuits}
to bound
\be
(1 - 2 \varepsilon)   \binom{d^{n} + t - 1}{t}  
\leq  |\text{supp}(\tilde \nu)| 
\leq  \binom{n}{2}^r  \left( \frac{20 rt}{\varepsilon}   \right)^{r d^4}.,
\ee
After some algebra, we obtain the desired bound on $r$.
which implies that
\be
r \geq \frac{nt}{2 \log(t) d^{4}}. 
\ee
\end{proof}

\section{Proof of Corollaries \ref{cor:hiding} and
  \ref{cor:followingcircuitsbyhaar}}\label{sec:hiding-proof} 

In the proof of \corref{hiding} we make use of the following lemma due
to Low, which gives a measure concentration result for $t$-designs.
Our definition of approximate $t$-designs differs from his by a
normalizing factor, and we have adjusted the statement of the
result accordingly.
\begin{lem}[Low, Theorem 1.2 of \cite{Low09}]  \label{Lowconcentration} Let $f : \mathbb{U}(D) \rightarrow \mathbb{R}$ be a polynomial of degree $K$. Let $f(U) = \sum_i \alpha_i M_i(U)$ where $M_i(U)$ are monomials and let $\alpha(f) = \sum_i |\alpha_i|$. Suppose that $f$ has probability concentration
\begin{equation}
\Pr_{U \sim \mu_{\text{Haar}}} \left( |f(U) - \mu| \geq \delta  \right) \leq C e^{-a \delta^{2}},
\end{equation}
and let $\nu$ be an $\epsilon$-approximate unitary $t$-design. Then for any integer $m$ with $2mK \leq t$,
\be
\Pr_{U \sim \nu} \left( |f - \mu| \geq \delta  \right) \leq \frac{1}{\delta^{2m}} \left( C \left( \frac{m}{a}  \right)^{m} + 2\eps (\alpha + |\mu|)^{2m}  \right).
\ee
\end{lem}

Let us turn to the proof of \corref{hiding}.

\restatecor{cor:hiding}
%\begin{replem}{cor:equilibration}\corhiding\end{replem}

\begin{proof} 
Consider a fixed POVM element $0 \leq M \leq I$. Let us apply Lemma
\ref{Lowconcentration} with $f_M(U) := \bra{0^{n}}U^{\cal y} M U
\ket{0^{n}}$. We have $D = d^{n}$, $K = 2$ and $\mu = \tr (M)/d^{n}\leq 1$. We can upper bound $\alpha(f)$ by $\sum_{i, j} |M_{i, j}| \leq d^{2n}$. Moreover, by Levy's lemma \cite{Ledoux},
\be
\Pr_{U \sim \mu_{\text{Haar}}} \left( |f_M(U) - \mu| \geq \delta  \right) \leq 2 e^{-\frac{d^{n} \delta^{2}}{140} }.
\ee
Choose
\be
 t = \L\lfloor \L(\frac{s}{1400 n^2 d^2 \log(d)}\R)^{1/11}\R\rfloor
\qquad\text{and}\qquad
\eps= \L(\frac{35t}{(d^{2n}+1)^2d^n}\R)^{t/4}.
\ee
From \corref{main-design} we get that local random quantum circuits of
size $s$ form an $\eps$-approximate unitary $t$-design.
Then using Lemma \ref{Lowconcentration} with $m = t/4$
\begin{eqnarray} \label{singletermbound}
\Pr_{U \sim \nu_{\text{LR}, d, n}^{*s} } \left( |f_{M}(U) - \mu| \geq \delta/2  \right) 
&\leq& 
\frac{1}{(\delta/2)^{t/2}} 
\left( 2 \left( \frac{140t}{4d^{n}}  \right)^{t/4} + 
2\eps (d^{2n}+1)^{t/2}  \right) \nonumber \\
&=& 4 \left( \frac{140t}{d^{n}\delta^2}  \right)^{t/4}.
\end{eqnarray}

Next, we let $S_r$ be a $\delta/2$-covering of the set of circuits of size $r$
(see \lemref{eps-net-circuits} for a definition).  By
\lemref{eps-net-circuits} and using $n\leq r$, we can assume
\be
|S_r| \leq \binom{n}{2}^r \left(\frac{20r}{\delta}\right)^{rd^4}
\leq \L(\frac{r}{\delta}\R)^{2rd^4}.
\ee

The quantity we are interesting in upper bounding is
\begin{eqnarray}
p:=\Pr_{U \sim \nu_{\text{LR}, d, n}^{*s} } \left(   \max_{M \hspace{0.05 cm} \in
    \hspace{0.05 cm} \text{size}(r)} |f_{M}(U) - \mu| \geq \delta
\right)
\leq 
\Pr_{U \sim \nu_{\text{LR}, d, n}^{*s}} \left(   \max_{M \in S_r} |f_{M}(U) - \mu| \geq \delta/2
\right)
\end{eqnarray}
Using a union bound, this latter probability is
\be
\leq \L(\frac{r}{\delta}\R)^{2rd^4}
\cdot 3 \left( \frac{560t}{d^{n}\delta^2}  \right)^{t/4}.
\ee
Substituting our choice of $t$, we see that this probability is
negligible when $\delta$ and $d$ are constant and $r\log(r) 
\ll s^{1/11}n^{9/11}$.
\end{proof}

Finally we prove

\begin{replem}{cor:followingcircuitsbyhaar}
\followingcircuitsbyhaar
\end{replem}

\begin{proof}
Let $-I \leq M \leq I$ and $\ket{\psi} \in \left(\mathbb{C}^d \otimes \mathbb{C}^d\right)^{\otimes n}$ be such that
\begin{eqnarray}  \label{bound}
X &:=& \L\|\int_{\bbU(d^n)} \ad_{C_U} \nu_{{\rm LR},d,n}^{*r}({\rm d}U) - \int_{\bbU(d^n)}
\ad_{C_U} \mu_{\rm Haar}({\rm d}U)\R\|_\diamond  \\ &=& \tr\left( M \left( \int_{\bbU(d^n)} (C_U \otimes I) \ket{\psi}\bra{\psi}(C_U \otimes I)^{\cal y} \nu_{{\rm LR},d,n}^{*r}({\rm d}U) - \int_{\bbU(d^n)} (C_U \otimes I) \ket{\psi}\bra{\psi}(C_U \otimes I)^{\cal y} \mu_{\rm Haar}({\rm d}U)    \right)  \right) \nonumber
\end{eqnarray}
Using repeatedly that $\tr((A_1 \otimes ... \otimes A_k)\mathbb{V}_{1, ..., k}) = \tr(A_1A_2...A_k)$, for $\mathbb{V}_{1, ..., k}$ a representation in $\left( \mathbb{C}^{d} \right)^{\otimes k}$ of a cycle, we can write 
\be
X = \tr \left( L \left( \int_{\bbU(d^n)} U^{\otimes t, t}\otimes I^{\otimes 2rn}\nu_{{\rm LR},d,n}^{*r}({\rm d}U)  -  \int_{\bbU(d^n)} U^{\otimes t, t}\otimes I^{\otimes 2rn}\mu_{\rm Haar}({\rm d}U)    \right)    \right),
\ee
with $t \leq r$ and $L := L_1...L_p$, where each $L_k$ is given by a tensor product of unitary operators and $\ket{\psi}\bra{\psi}$. Thus $\Vert L \Vert_{1} \leq  d^{4rn} \Vert L \Vert_{\infty} \leq d^{4rn}$ and so
\begin{eqnarray}
X &\leq& \Vert L \Vert_{1} \left \Vert    \int_{\bbU(d^n)} U^{\otimes t, t} \otimes I^{\otimes 2rn}\nu_{{\rm LR},d,n}^{*r}({\rm d}U)  -  \int_{\bbU(d^n)} U^{\otimes t, t} \otimes I^{\otimes 2rn}\mu_{\rm Haar}({\rm d}U)     \right \Vert_{\infty} \nonumber \\ &\leq&  d^{4rn}  \left \Vert    \int_{\bbU(d^n)} U^{\otimes t, t}\nu_{{\rm LR},d,n}^{*r}({\rm d}U)  -  \int_{\bbU(d^n)} U^{\otimes t, t}\mu_{\rm Haar}({\rm d}U)     \right \Vert_{\infty}.
\end{eqnarray}
The statement follows from part 1 of Theorem \ref{thm:main-TPE}.
\end{proof}

\section{Proof of Corollary \ref{cor:topologicalorder}}\label{sec:topo-order}

%\restatecor{cor:topologicalorder}
\begin{repcor}{cor:topologicalorder} \topologicalorder{second} \end{repcor}

\begin{proof}
A standard estimate gives that for a Haar random unitary $U$:
\begin{equation}
\mathbb{E}_{U \sim \mu_{\text{Haar}}} \left \Vert \tr_{\backslash X} \left( U \ket{\psi_0}\bra{\psi_0}U^{\cal y} \right) - \tau_{X}  \right \Vert_2^2 \leq 2^{-(n - l)}.
\end{equation}
Apply part 2 of Theorem \ref{thm:main-TPE} with $d=2,t=2,\eps=2^{-n}$
to find that
\begin{equation}
\mathbb{E}_{U \sim \left(\nu_{\text{PLR}, d, n} \right)^{* C n} }
\left \Vert \tr_{\backslash X} \left( U
    \ket{\psi_0}\bra{\psi_0}U^{\cal y} \right) - \tau_{X}  \right
\Vert_2^2 \leq 2^{-(n - l)} + 2^{-n}, 
\end{equation}
for $C \leq 10^9$.
\eq{tqo1-second} then follows from the relation $\Vert Z \Vert_{1} \leq \sqrt{D} \Vert Z \Vert_2$, valid for any $D \times D$ matrix $Z$.

The proof of \eq{tqo2-second} is completely analogous, we just have to note that 
\begin{equation}
\mathbb{E}_{U \sim \mu_{\text{Haar}}} 
\left \Vert \tr_{\backslash X} \left( U
    \ket{\psi_0}\bra{\psi_1}U^{\cal y} 
\right)  \right \Vert_2^2 \leq 4\cdot 2^{-(n - l)}.
\end{equation}
\end{proof}

%The idea is that if $\nu$ is a $t$-design, then we can replace the
%part of it consisting of circuits of size $\leq r$ with a distribution
%over the $\eps$-covering from \lemref{eps-net-circuits}.

%Dobrushin-Schlosman ergodicity condition

%\section{Conclusions and Open Questions}

\section{Discussion and open questions}
Our work shows that random unitary circuits resemble Haar-uniform
unitaries in the sense of being approximate $k$ designs.  Of course,
this is not the only possible criteria for rapid mixing.  At least two
other conditions that would be interesting to investigate are the
rapid-scrambling criterion~\cite{SS08, LSHOH11} and the log-Sobolev
condition~\cite{KT12}.  In general, it makes sense to investigate
these conditions in the context of an application, and here too more
work could be done to clarify questions such as which definition of
$\eps$-approximate $t$-designs is most natural.

Another natural open question is to extend our results to an arbitrary
universal set of gates.  Currently we can only handle 
the case of sets composed of gates with algebraic entries, but we
believe this is just a technicality.

For nearest-neighbor circuits in one dimension, our results are nearly
optimal.  In other geometries, random circuits mix at least as well,
but in this case, our results are likely to be far from optimal.  For
interactions on general graphs, it is plausible that parallel random
circuits mix in time comparable to the diameter of the graph, which we
establish only in the case of graphs of linear diameter.

Another open question is whether our results can be
strengthened to prove the incompressibility of quantum circuits.  See
Application II in \secref{app} for more discussion of this point.

Finally, in physical systems, it is natural to consider random
time-independent {\em Hamiltonians}, rather than random sequences of
unitaries.  Here the situation is qualitatively unlike that of
classical time-independent stochastic processes, since the phenomenon
of Anderson localization prevents mixing in many cases (see
\cite{Stolz11} for a review).  It is an intriguing open question to
understand the cases in which rapid mixing nevertheless occurs, and in
particular to give a physically plausible derivation of the observed
phenomenon of thermalization.

\section{Acknowledgments}

We would like to thank Dorit Aharonov, Itai Arad, Winton Brown, Daniel
Jonathan, Emilio Onorati, Alex Russell, and Tomasz Szarek for useful discussions and
suggestions on approximate unitary designs, Bruno Nachtergaele for
helpful explanations about \cite{Nac96}, Kevin Zatloukal for sharing
his results about partitions and Andreas Winter for suggesting
\defref{design}. We are grateful to Juani Bermejo Vega, Dominik Hangleiter and Jonas
Haferkamp  for making us aware of mistakes in the previous version.
We would like to thank the anonymous referee for many useful comments. 
FGSLB acknowledges support from EPSRC through an Early Career 
Fellowship, the Swiss National Science Foundation, via the National Centre of Competence in Research
QSIT, the Ministry of Education and the National Research Foundation, 
Singapore, and the Polish Ministry of Science and Higher Education Grant No. IdP2011 000361. 
AWH was funded by NSF grants 0916400, 0829937, 0803478, 1111382,
1452616; DARPA
QuEST contract FA9550-09-1-0044; and the IARPA MUSIQC and QCS
contracts. MH is supported by EU grant QESSENCE, the project QUASAR of
the National Centre for Research and Development of Poland and by
Polish Ministry of Science and Higher Education grant N N202
231937. Part of this work was done in National Quantum Information
Center of Gdansk. We thank the hospitality of Institute Mittag Leffler
within the program ``Quantum Information Science'', where (another)
part of this work was done.

\appendix

\section{Facts about designs and tensor-product expanders}
\label{app:design-facts}
\begin{definition}\label{def:sym}
Define the {\em symmetric subspace} of $(\bbC^N)^{\ot t}$ to be the set
of vectors that are invariant under $V_N(\pi)$ (defined in
Eq. \eq{Vd-def-first}) for all $\pi\in \cS_t$.
Denote this subspace by $\vee^t\bbC^N$.  Note that $\dim\vee^t\bbC^N =
\binom{N+t-1}{t}$. 
\end{definition}

Before proving \lemref{design-defs}, we relate the design condition of
\defref{design} to an easier-to-prove condition.

\begin{lem}\label{lem:design-J}
For $\nu$ a measure on $\bbU(N)$, define:
\begin{subequations}\label{eq:Delta-J-defs}
\ba \rho_\nu &:= (\Delta_{\nu, t}\ot\id)(\Phi_N^{\ot t}) \\
\rho_{\text{Haar}}& := (\Delta_{\mu_{\text{Haar}},
  t}\ot\id)(\Phi_N^{\ot t})
\ea \end{subequations}
Then 
\benum
\item The support of $\rho_\nu$ and $\rho_{\text{Haar}}$ is contained in
  $S:=\vee^t(\bbC^N \ot \bbC^N)$.
\item The minimum eigenvalue of $\rho_{\text{Haar}}|_S$ is $\geq N^{-2t}$.
\item 
\be G(\nu,t) \leq N^{2t} \|\rho_\nu - \rho_{\text{Haar}}\|_\infty
\label{eq:design-from-J}\ee
\eenum
\end{lem}

\begin{proof}
\benum
\item
Each density matrix is a mixture of states of the form $((U \ot I)\ket{\Phi_N})^{\ot
  t}$, each of which individually belongs to $\vee^t(\bbC^N \ot
\bbC^N)$.
\item
We will use Schur duality (see \cite{GW98,Har05} for
reviews).   Schur duality implies that
\be \ket{\Phi_N}^{\ot t} \cong
\sum_{\lambda \in \Par(t,N)} \sqrt{\frac{\dim\cQ_\lambda^N
    \dim\cP_\lambda}{N^t}}
\ket{\lambda,\lambda}\ot\ket{\Phi_{\cQ_\lambda^N}}
\ot\ket{\Phi_{\cP_\lambda}},\ee
where $\Par(t,N)$ denotes the partitions of $t$ into $\leq N$ parts,
$\cQ_\lambda^N$ denotes the irrep of $\bbU(N)$ corresponding to
partition $\lambda$, 
$\cP_\lambda$ the irrep of $\cS_t$ corresponding to $\lambda$, and
$\ket{\Phi_{\cQ_\lambda^N}}$, $\ket{\Phi_{\cP_\lambda}}$ refers to
maximally entangled states on pairs of these spaces.
Then 
\ba \rho_{\text{Haar}} 
&= \int_{\bbU(N)} (\ad_{U \ot I}(\Phi_N))^{\ot
  t} \mu_{\text{Haar}}({\rm d}U) \\
& \cong 
\sum_{\lambda \in \Par(t,N)} \frac{\dim\cQ_\lambda^N
    \dim\cP_\lambda}{N^t}
\proj{\lambda}^{\ot 2}\ot 
\L(\frac{I_{\cQ_\lambda^N}}{\dim \cQ_\lambda^N}\R)^{\ot 2}
\ot \proj{\Phi_{\cP_\lambda}}.
\ea
Restricting to $\vee^t(\bbC^N \ot\bbC^N)$, we find that the minimum
eigenvalue is $\min_\lambda
\frac{\dim\cP_\lambda}{N^t \dim\cQ_\lambda^N}$.  This minimum is
achieved by the symmetric irrep $\lambda = (t)$, for which
$\dim\cP_\lambda=1$ and $\dim\cQ_\lambda^N = \binom{N+t-1}{t} \leq
N^t$.  Thus $\lambda_{\min}(\rho_{\text{Haar}}|_S) \geq N^{-2t}$.
\item This follows from parts 1 and 2 of this Lemma, along with \eq{design-def-J}.
\eenum
\end{proof}

\begin{proof}[Proof of \lemref{design-defs}]
Let $\Theta := \Delta_{\nu,t} - \Delta_{\mu_{\text{Haar}},t}$.
Then
\ba
 \|\Theta\|_\diamond &= \max_\psi \|(\Theta \ot \id)(\psi)\|_1 \\ 
&= \max_{0\leq M_-,M_+ \leq I} \max_\psi
\L[\tr M_+(\Theta \ot \id)(\psi) - \tr M_-(\Theta \ot \id)(\psi)\R] \\
&\leq 2\eps,
\ea
where the last line used \eq{design-def}.

Conversely, suppose that 
\be \| \Delta_{\mu_{\text{Haar}, t}} -\Delta_{\nu, t}\|_\diamond \leq
\eps \label{eq:diamond-bound}\ee
Define $\rho_\nu,\rho_{\text{Haar}}$ as in \eq{Delta-J-defs}.
Then by \eq{diamond-bound} we have 
\be \eps \geq 
\|\rho_\nu -\rho_{\text{Haar}}\|_1  \geq
\|\rho_\nu -\rho_{\text{Haar}}\|_\infty 
\ee
The desired claim now follows from \lemref{design-J}.
\end{proof} 

\begin{proof}[Proof of \lemref{design-expander}]
For the first inequality, observe that
\be
g(\nu, t) \stackrel{(1)}{\leq} N^{t/2} \| \Delta_{\nu, t} -
\Delta_{\mu_{\text{Haar}}, t}\|_\diamond \stackrel{(2)}{\leq} 2N^{t/2} G(\nu, t),
\ee
Here (1) is from part of Lemma 2.2.14 of \cite{Low10} (see in particular
Fig 2.1 of \cite{Low10}), with the
definition OPERATOR-2-NORM from \cite{Low10} corresponding to the TPE
condition here, and (2) is from \lemref{design-defs}.

For the second inequality, we use the fact that the $2\rar 2$ norm is
stable under tensoring with the identity 
map to obtain
\be \| (\Delta_{\nu, t} - \Delta_{\mu_{\text{Haar}}, t}) \ot \id_{N^t}
\|_{2\rar 2} \leq g(\nu,t).
\ee
Thus, defining $\rho_\nu, \rho_{\text{Haar}}$ as in \eq{Delta-J-defs},
we have
\be g(\nu, t)\geq \|\rho_\nu - \rho_{\text{Haar}}\|_2 
\geq \|\rho_\nu - \rho_{\text{Haar}}\|_\infty.
\ee
Thus, \lemref{design-J} implies that $G(\nu, t) \leq N^{2t}g(\nu, t)$.
\end{proof}

%\bibliographystyle{hyperabbrv}
%\bibliography{ref}

\end{document}